\documentclass[fleqn,reqno]{article}
\usepackage{amsmath}
\usepackage{amsthm}
\usepackage{amssymb,bbm,graphicx}

\newtheorem{thm}{Theorem}[section]
\newtheorem{proposition}[thm]{Proposition}
\newtheorem{remark}[thm]{Remark}
\newtheorem{lemma}[thm]{Lemma}
\newtheorem{prop}[thm]{Proposition}

\newcommand{\R}{\mathbb{R}}

\newcommand{\RR}{\mathbb{R}}

\def\11{\mathbbmss{1}}

\newcommand{\e}{\mathrm{e}}
\renewcommand{\d}{\mathrm{d}}

\renewcommand{\Re}{\operatorname{Re}\,}

\newcommand{\eps}{\varepsilon}

\newcommand{\DETAILS}[1]{}

  \newcommand{\Rp}{{\mathbb R}_+}

\newcommand{\cL}{\mathcal{L}}         
\newcommand{\cM}{\mathcal{M}}         

\newcommand{\g}{{\gamma}}

\newcommand{\lam}{{\lambda}}           

\newcommand{\s}{{\sigma}}
\newcommand{\z}{{\zeta}}

\renewcommand{\d}{\mathrm{d}}

\newcommand{\ran}{\rangle}
\newcommand{\lan}{\langle}
\newcommand{\ra}{\rightarrow}

 \newcommand{\p}{{\partial}}

\newcommand{\grad}{\mbox{grad}\, }

\newcommand{\lb}{\big(}
\newcommand{\rb}{\big)}
\newcommand{\lsb}{\big[}
\newcommand{\rsb}{\big]}

\newcommand{\ls}{\lesssim}

\newcommand{\Ef}{{\mathcal E}_f}
\newcommand{\Er}{{\mathcal E}}

\newcommand{\scalar}[2]{\langle{#1} \mspace{2mu}, {#2}\rangle}
\newcommand{\norm}[1]{\lVert #1 \rVert}
\newcommand{\avg}[1]{\langle #1 \rangle}

\newcommand{\LpNorm}[2]{\big\|#2\big\|}
\newcommand{\HsNorm}[2]{\left\|#2\right\|_{H^{#1}}}

\newcommand{\ip}[2]{\left\langle #1,#2\right\rangle}
\newcommand{\ipp}[2]{( #1,#2 )}

\newcommand{\Question}[1]{}
\newcommand{\SID}[1]{}

 \renewcommand{\O}[1]{\mathrm{O}\lb#1\rb}
 \newcommand{\smallO}[1]{\mathrm{o}\lb#1\rb}

\newcommand{\zic}{\zeta_{bci}}

\newcommand{\const}{\operatorname{const}}
\newcommand{\Lap}[1]{\Delta^{\!\!(#1)}}

\newcommand{\Lab}{{\cal L}_{ab}}
\newcommand{\Fab}{{\cal F}_{ab}}
\newcommand{\Nab}{{\cal N}}

\newcommand{\Labc}{{\cal L}_{abc}}
\newcommand{\Fabc}{{\cal F}_{abc}}

\newcommand{\DATUM}{February 10, 2013}  
\title{On blowup dynamics in the Keller-Segel model of chemotaxis}

\author{S. I. Dejak\thanks{Department of Mathematics, University of Toronto, Toronto, Canada.}\ \qquad \qquad
D. Egli\thanks{Department of Mathematics University of Toronto, Toronto, Canada.}\ \qquad \qquad
P.M. Lushnikov\thanks{Department of Mathematics and Statistics, University of New Mexico, USA}\ \qquad \qquad
I. M. Sigal\thanks{Department of Mathematics University of Toronto, Toronto, Canada.}}

\date{\DATUM}
\begin{document}


\maketitle

  \centerline{\it In memory of V.S. Buslaev, a scientist and a friend}

  \centerline{\it  Will appear in St.Petersburg Math Journal (issue dedicated to V.S. Buslaev)" }

\bigskip

\begin{abstract}
We investigate the  (reduced) Keller-Segel equations modeling chemotaxis of bio-organisms. We present a formal derivation and partial rigorous results of the blowup dynamics of solution of these equations describing   the chemotactic aggregation of the organisms. Our results are confirmed by numerical simulations and the formula we derive coincides with the formula of Herrero and Vel\'{a}zquez for specially constructed solutions.
\end{abstract}

\section{Introduction}
In this paper we 
analyze the aggregation dynamics in the (reduced) Keller-Segel model of chemotaxis.
Chemotaxis is the directed movement of organisms in response to the concentration gradient of an external chemical signal and is common in biology.  The chemical signals can come from external sources or they can be secreted by the organisms themselves.  

Chemotaxis is believed to underly many social activities of micro-organisms, e.g.\ social motility, fruiting body development, quorum sensing and biofilm formation. A classical example
is the dynamics and the aggregation of {\it Escherichia coli} colonies under starvation conditions \cite{BrLeBu1998}. 
Another example is the   {\it Dictyostelium} amoeba , where single cell bacterivores,
 when challenged by adverse conditions, form multicellular structures of $\sim 10^5$ cells \cite{Bo1967,CFTV}.  Also, endothelial cells of humans react to vascular endothelial growth factor to form blood vessels
 through aggregation \cite{CarmelietNatMEd2000}.

Consider organisms moving and interacting in a domain $\Omega\subseteq \R^d$, $d=1,2$ or $3$.  Assuming that the organism population is large and the individuals are small relative to the domain $\Omega$, Keller and Segel derived a system of reaction-diffusion equations governing the organism density $\rho:\Omega\times\Rp\rightarrow\Rp$ and chemical concentration $c:\Omega\times\Rp\rightarrow\Rp$.  The equations are of the form
\begin{equation}\label{KS}
\begin{split}
\p_t \rho &=D_\rho\Delta\rho-\nabla\cdot\lb f(\rho)\nabla c\rb\\
\p_t c&=D_c\Delta c+\alpha \rho-\beta c.
\end{split}
\end{equation}
Here $D_\rho$, $D_c$, $\alpha$, $\beta$ are positive functions of $x$ and $t$, $\rho$ and $c$, and $f(\rho)$ is a positive function modeling chemotaxis.  Assuming a closed system, one is led to impose no-flux boundary conditions on $\rho$ and $c$:
\begin{equation}
\p_\nu \rho=0\ \mbox{and}\ \p_\nu c=0\
\mbox{on}\ \p\Omega,
\label{eqn:KS1BC}
\end{equation}
where $\p_\nu g$ is the normal derivative of $g$.  With these boundary conditions, the total number of organisms in $\Omega$ is conserved. 
We refer the reader to \cite{Bo1967, BrLeBu1998, KeSe1970, Na1973} for more information about chemotaxis and the Keller-Segel model.

We presently concentrate on the case of positive chemotaxis,
where the organisms secrete the chemical and move towards areas of higher chemical concentration.
This leads to aggregation of organisms. 
 Mathematically this is expressed as a  blowup (or collapse) of solutions of \eqref{KS}. 
It was first suggested by Nanjundiah in \cite{Na1973} that the density, $\rho$, may become infinite and form a Dirac delta singularity.  One refers to this process as (chemotactic) collapse.
This is, arguably,  the most interesting feature of  the  Keller-Segel equations.
As argued below, the ``collapsing'' profile and contraction law have a universal (close to self-similar)
 form, independent of particulars of initial configurations and, to a certain degree, of the equations themselves, and can be associated with chemotactic aggregation.
 Though the equations are rather crude and unlikely to produce  patterns one observes in nature or experiments, the collapse phenomenon could be useful in  verifying
 assumptions about biological mechanisms.\footnote{There are numerous refinements of the Keller-Segel equations, e.g.\ taking into account finite size of organisms
 (\cite{AlberChenGlimmLushnikov2006,AlberChenLushnikovNewman2007,LushnikovChenalberPRE2008}) preventing complete collapse, which model chemotaxis more precisely. We believe the techniques we outline and develop here can be applied to these models as well.}

 Phenomena of blowup and collapse in nonlinear evolution equations 
are hard to simulate numerically and the rigorous theory, or at least a careful analysis, is  pertinent here.
The recent years witnessed a tremendous progress in the development of such theories. 
We can now describe  the shape of blowup profile and contraction law 
in Yang-Mills, $\s-$model,  nonlinear Schr\"odinger and heat  equations (\cite{RS, RR, KST1, KST2, OS, BOS, MR, MZ, DGSW})
\footnote{Numerical simulations for these equations failed until the compression rate was derived analytically, see \cite{BOS, OS, SS}}.
Yet, after 40 years of intensive research and important progress, we still cannot give a rigorous description of collapse in the Keller-Segel equations modeling chemotaxis.
(See \cite{BrLeBu1998,Ve1, Ve2, BCC, BCL, BCM, BDEF, BDP, BeCL} for some recent works, \cite{BrCoKaScVe1999}, for a nice discussion of the subject,  
and \cite{Na2001, Ho2003, Ho2004, HP, Per} for reviews.)

This is not to say that the Keller-Segel equations are harder than Yang-Mills, $\s-$model, or nonlinear Schr\"odinger equations, they are not, but neither are they less important.

There are three common approximations made in the literature for system \eqref{KS}.  Firstly, one assumes that the coefficients in \eqref{KS} are constant and satisfy
\begin{equation}
\begin{split}
\epsilon:=\frac{D_\rho}{D_c}\ll 1,\ \tilde{\alpha}:=\frac{\alpha}{D_c}=\O{1}\ \mbox{and}\ \tilde{\beta}:=\frac{\beta}{D_c}\ll 1.
\end{split}
\label{approxrelations}
\end{equation}
The first of these conditions states that the chemical diffuses much faster than the organisms do.  This is the case in practically all situations.  As a result of this relation, 
one drops the $\p_t c$ term in \eqref{KS} (after rescaling time $t\rightarrow t/D_\rho$, this term becomes $\epsilon\p_t c$).  Secondly, one takes $f(\rho)$ to be a linear function $f(\rho)=K\rho$.
Thirdly,  the term $\beta c$ in \eqref{KS} is neglected compared with $\alpha\rho$, as one expects that it would not effect the blow-up process where $\rho\gg 1$ (it is also small due to the last relation in \eqref{approxrelations}). 
These approximations,  after rescaling, lead to the system
 \begin{equation}\label{KS3} \begin{split}
\frac{\p \rho}{\p t}&=\Delta \rho-\nabla\cdot\lb\rho \nabla c\rb,\\
0&=\Delta c+\rho ,
\end{split}
\end{equation}
with $\rho$ and $c$ satisfying the no-flux Neumann boundary conditions.

 Equations \eqref{KS3} 
   in three dimensions also appear in the context of stellar collapse  (see \cite{HeMeVe1997, Wo1992, ChavSir, SirChav});  similar equations---the Smoluchowski or nonlinear Fokker-Planck equations---models non-Newtonian
complex fluids (see \cite{Doi, Lar, CKT1, CKT2}. This is the equation studied in this paper.

We emphasize that in dropping the time derivative term of $c$, we have made the adiabatic approximation, in which the chemical is assumed to reach its steady state given by the second equation of \eqref{KS3} instantaneously.

In this paper, we consider the collapse of radially symmetric solutions to the reduced Keller-Segel system \eqref{KS3} on the plane $\R^2$ with a smooth, positive and integrable initial condition $\rho_0$ 
and with the boundary conditions $\rho,\nabla\rho,\nabla c\rightarrow 0$ as $|x|\rightarrow\infty$.
To provide a right context for the discussion below, we mention that equation \eqref{KS3} has the following key properties:
\begin{itemize}
\item It is invariant under the scaling transformations
\begin{align}\label{scaling}
\rho(x,t)\rightarrow\frac{1}{\lambda^2}\rho\lb\frac{1}{\lambda}x,\frac{1}{\lambda^2}t\rb\
\mbox{and}\ c(x,t)\rightarrow c\lb\frac{1}{\lambda} x,\frac{1}{\lambda^2} t\rb.
\end{align}
\item
It has the static solution,
\begin{equation} \label{stat-sol}
R(x):=\frac{8}{(1+|x|^2)^2},\ C(x):=-2\ln(1+|x|^2).\end{equation}
\item  The total ``mass'' is conserved:
$ \int_\Omega \rho(x,t)\, dx=\int_\Omega \rho(x,0)\, dx.$ 
\end{itemize}
We also mention that \eqref{KS3} (as well as \eqref{KS} ) is a gradient flow, $\p_t\rho=\nabla\cdot\rho\nabla\Er' (\rho)$, or $\p_t\rho=-\mathrm{grad}\, \Er(\rho)$, where $\Er' (\rho)$ is the formal $L^2-$gradient of $\Er$ and $\mathrm{grad}\, \Er(\rho)$ is the formal gradient of $\Er$ in the space with metric
$\ip{v}{w}_J:=-\ip{v}{J^{-1} w}_{L^2}$. Here $J:=\nabla\cdot\rho\nabla\le 0$, whose inverse is unbounded operator, and
$\Er(\rho)$  is the ``energy'' functional given by
\begin{align}\label{energy}\Er(\rho)&= \int_{\R^2} (-\frac{1}{2}\rho \Delta^{-1}\rho+\rho\ln\rho-\rho)\, dx  
\end{align}
(see Appendix \ref{sec:grad} for more details).
We remark that the first term of $\Er$ can be thought of as the internal energy of the system and the remaining terms are the entropy. The solution \eqref{stat-sol} is a minimizer of ${\cal E}$ under the constraint that $\int\rho=\const.$
Note that $\int_{\R^2} R\, dx=8\pi$, which is the source of $8\pi$ in \eqref{blowupcrit}.
Under the scaling \eqref{scaling}, the total mass changes as
\begin{equation*}
\int\frac{1}{\lambda^2}\rho\lb\frac{1}{\lambda}x,0\rb=\lambda^{(d-2)}\int\rho\lb
x,0\rb.
\end{equation*}
Thus one does not expect collapse for $d=1$, and that collapse
is possible for $d\geq 2$ with critical collapse for $d=2$
and supercritical collapse for $d>2$. (Equation \eqref{KS3} in $d=2$ is said to be $L^1-$critical, etc.)

Take $\rho_0\ge 0$. One has the following criteria for blowup of solutions of \eqref{KS3} (\cite{Na1995, Bi1998}):
 If the dimension $d=2$ and the total mass satisfies
\begin{equation} \label{blowupcrit}
M:=\int_{\R^2} \rho_0\, dx>8\pi,
\end{equation}
or, if $d\ge 3$ and
$\frac{\int_{\R^d}x^2\rho_0dx}{\int_{\R^d}\rho_0dx}$
is sufficiently small (this means that $\rho_0$ is concentrated at
$x=0$), then the solution to \eqref{KS3} blows up in finite time.

There is a fair amount of work done on equations \eqref{KS} 
 and \eqref{KS3} and closely related equations. We give a very brief and incomplete review of it.
Childress and Percuss \cite{ChPe1981} found that collapse for \eqref{KS} with $f$ linear does not occur when $d=1$ and can occur when $d\ge 3$.  For the two-dimensional case, they advanced arguments that collapse requires a threshold number of organisms.  This threshold behaviour was confirmed by J\"{a}ger and Luckhaus in \cite{JaLu1992}
\DETAILS{ for systems of the form \eqref{KS2}.  They proved that there exists a constant $c_1$, depending on the bounded, $C^1$ domain $\Omega$, such that, if $\bar{\rho_0}:=|\Omega|^{-1}\int_\Omega \rho_0 < c_1$, then there exists a unique, smooth global solution $\rho$ to \eqref{KS2}. \textbf{???} On the other hand, if $\Omega$ is a disk, then there exists another constant $c_2=c_2(\Omega)$ such that if $\bar{\rho_0}>c_2$, then there are radially symmetric solutions that blowup at the origin in finite time:
\begin{equation*}
 \lim_{t\rightarrow T} \rho(0,t)=\infty,
\end{equation*}
for some $0<T<\infty$.  It was later proved that the two thresholds, $c_1$ and $c_2$, are both equal to $8|\Omega|$ \textbf{???}}
(see also \cite{Na1995, NaSe1998, NaSe1997, NaSeYo1997, NaSeYo1997, NaSeSu2000,Na2001}). 

Herrero and Vel\'{a}zquez proved that there exist radial solutions of \eqref{KS3} for $d=2$ with the threshold mass $8|\Omega|$ collapsing to a Dirac delta singularity in finite time (see \cite{HeVe1996a}).  Also, unlike previous results, the authors give an explicit asymptotic expression of the developing singularity.  They proved using matched asymptotics and a topological argument that for $T>0$ there exists a radial solution to \eqref{KS3}, 
which blows up at $r=0$ and $t=T$ and is of the form
\begin{equation*}
 \rho(r,t)=\frac{1}{\lambda(t)^2}R_{\lambda(t)} (1+\smallO{1})+
\left\{
\begin{array}{ll}
0& r<\lambda(t)\\
\O{\frac{e^{-\sqrt{2}|\ln(T-t)|^\frac{1}{2}}}{r^2}} & r\ge\lambda(t),
\end{array}\right.
\end{equation*}
as $t\rightarrow T$, where $R_{\lambda}(r):= R(r/\lambda),\ R(r)$ is the stationary solution to \eqref{KS3} (see \eqref{stat-sol}) 
and
\begin{equation*}
\lambda(t)=C(T-t)^\frac{1}{2} e^{-\frac{1}{\sqrt{2}}|\ln(T-t)|^\frac{1}{2}}|\ln(T-t)|^{\frac{1}{4}|\ln(T-t)|^{-\frac{1}{2}}-\frac{1}{4}}(1+\smallO{1}).
\end{equation*}
They also considered collapse of solutions to \eqref{KS} with linear $f(\rho)$ (see \cite{HeVe1996b} and \cite{HeVe1997}).  Obtaining similar results, they suggest that J\"{a}ger and Luckhaus' adiabatic assumption does not affect the collapse mechanism.
 In the papers \cite{Lush, DLV} Lushnikov et al derived the  log-log scaling as well as corrections  beyond leading order log-log scaling.

 As noted in \cite{HeVe1996a}, the asymptotics reproduced above are not of self-similiar type; that is, they are not of the form $(T-t)\Phi(r/(T-t)^\frac{1}{2})$ for some function $\Phi$.  In fact, as shown in \cite{HeMeVe1998}, self-similiar blowup is not possible.  Lastly, we mention that similar work has been done for the three dimensional case, where existence of collapsing shock waves has been shown.  We refer the reader to \cite{HeMeVe1998, HeMeVe1997, BrCoKaScVe1999} for these results.

The above results are valid for radially symmetric  domains and initial conditions.  It was shown by numerous authors that the blowup threshold mentioned above decreases for the non-spherically symmetric situation.  Moreover, Dirac delta singularities may develop on the boundary of the domain.  We refer the reader to \cite{Bi1995, Bi1998,GaZa1998,Ho2001,HoWa2001, Na20012} for details. In \cite{Ve1}, Vel\'{a}zquez considers small radial and non radial perturbations of a collapsing solution and concludes, using formal matched asymptotics, that they are stable to these perturbations, leading only to small shifts in the blowup time and the blowup point.  Existence of blowup or bounded solutions when $f(\rho)$ is nonlinear was recently studied in \cite{HoWi2005}.  We also refer to \cite{Ho2002} for a blowup result of a related Keller-Segel model.  Lastly, we refer the reader to Horstmann \cite{Ho2003, Ho2004} for a more complete review of the literature including results on other models of chemotaxis and on the derivation of the Keller-Segel model as a continuous limit of biased random walks (see e.g. \cite{Oe1989, OtSt1997, St2000}).

In spite of the considerable progress,  the question of whether 
the mass collects in isolated points, forming Dirac delta distributions, remained unanswered.  Moreover, these results give no information about the dynamics of blowup. These are the questions we address.

Now, we describe the results of the present paper. Given a radially symmetric initial condition $\rho_0(r)>0$ sufficiently close to some $R_{\lambda_0}$, for some $\lambda_0$, and satisfying $\int\rho_0>\int R$, we show (formally, but with some rigorous supporting results) that the solution $\rho(x,t)$ to \eqref{KS3} is of the form
\begin{equation*}
\rho(x,t)=\frac{1}{\lambda^2(t)}R_{\lambda(t)}(r)(1+\smallO{1})
\end{equation*}
with $\lambda(t)\rightarrow 0$ as $t\rightarrow T$ for some $0<T<\infty$.  Thus, all the mass $\int\rho\, dx$ collapses to the single point $x=0$ in finite time, or equivalently, the density $\rho$ forms a Dirac delta singularity with weight $8\pi$ in finite time. Furthermore, we show that the compression scale, $\lambda$,
has the following explicit asymptotics
 \begin{equation} \label{lambdaAs}
\lambda(t)=c(T-t)^\frac{1}{2}e^{-\frac{1}{\sqrt{2}}|\ln(T-t)|^\frac{1}{2}}|\ln(T-t)|^{\frac{1}{4}
}(1+\smallO{1})
\end{equation}
for some constant $c$. In Figure \ref{fig1} we compare the blowup asymptotics
 \eqref{lambdaAs} with direct numerical simulation of \eqref{KS3}.

  We also give an estimate of the error term, $\rho(x,t)-R_{\lambda(t)}(r)$, in the case when the nonlinear part in equation \eqref{phi-eq'}, given below, can be neglected.
We believe that our results and our analysis can be made rigorous and can be extended to the full Keller-Segel system.
 \begin{figure}[ht]
 \includegraphics[width=5in]{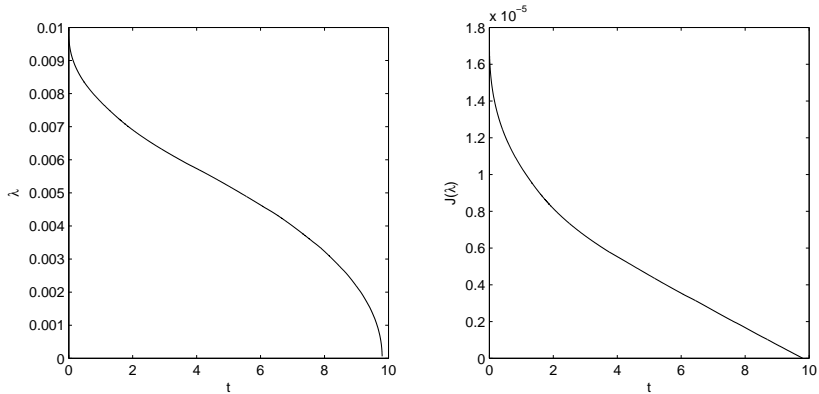}
\caption{The left pane shows the scaling parameter $\lambda(t)$ obtained by numerically computing the solution to \eqref{KS3} with the initial condition $m_0:=4 y^2/(1+\delta y+y^2)|_{y=r/\lambda_0}$ with $\lambda_0=-\delta=.01$.  The right pane plots the quantity $J(\lambda):=  \dfrac{e^{\sqrt{4\ln \frac{\lambda_0}{\lambda}}}}{\sqrt{\ln \frac{\lambda_0}{\lambda}}}(\frac{\lambda}{\lambda_0})^2
 $ against time, which according to \eqref{lambdaAs} 
should be linear as the blowup time is approached.}
\label{fig1}
 \end{figure}

We outline the approach used in this paper.  In the case of  radially symmetric solutions,  the system  \eqref{KS3}, which consists of  coupled parabolic and elliptic PDEs, 
is equivalent to a single PDE. Indeed, the change of the unknown, by passing from the density, $\rho(x,t)$, to the normalized mass, 
\begin{equation*}
m(r,t):=\frac{1}{2\pi}\int_{|x|\le r}\rho(x,t)\ dx,
\end{equation*}
  of organisms contained in a ball of radius $r$, 
discovered by  \cite{JaLu1992, BrCoKaScVe1999}, maps two equations \eqref{KS3} into a single equation  
\begin{equation}\label{m-eq}
\p_t m=\Lap{0}_r m+r^{-1} m\p_r m,
\end{equation}
on $(0,\infty)$ (with initial condition $m_0(r):=\frac{1}{2\pi}\int_{|x|\le r}\rho_0 (x)\, dx$).  Here $\Lap{n}_r$ is the $n$-dimensional radial Laplacian, $\Lap{n}_r:=r^{-(n-1)}\p_r r^{n-1}\p_r=\p_r^2+\frac{n-1}{r}\p_r$. 
Thus \eqref{KS3} in the radially symmetric case is equivalent to \eqref{m-eq} and therefore we concentrate on the latter equation.

The properties of equation \eqref{KS3} discussed above imply the following key properties  equation of \eqref{m-eq}
\begin{itemize}
\item It is invariant under the scaling transformations
$m(r,t)\rightarrow m\lb\frac{1}{\lambda}r,\frac{1}{\lambda^2}t\rb.$ 
\item
It has the static solution (coming from the static solution $R(r)=\frac{8}{(1+r^2)^2}$ of \eqref{KS3}),
\begin{equation}\label{defn:chi} \chi(r):=\frac{4 r^2}{1+r^2}.
\end{equation}

\item  The total ``mass'' is conserved:
$2\pi \lim_{r\ra\infty} m(r,t)=\int \rho(x, t)dx = \const.$ 
\end{itemize}
Note that the stationary solution  has total mass $2\pi \lim_{r\ra\infty} \chi(r)=8\pi$, which,  recall,  is the sharp threshold between global existence and singularity development in solutions to \eqref{KS3} (see \eqref{blowupcrit}). 

The properties above yield, as in the case of \eqref{KS3}, the manifold of static solutions 
${\cal M}_0:=\{\chi(r/\lambda)\ |\ \lambda>0\}$ 
and suggest a likely scenario of collapse: sliding along ${\cal M}_0$ in the direction of $\lambda \ra 0$.
To analyze the collapse, we  pass to the reference frame collapsing with the solution, by introducing the adaptive blowup variables, 
\begin{equation*}
m(r,t)=u (y,\tau),\ \quad \mbox{where}\ \quad y=\frac{r}{\lambda}\
\mbox{and}\ \tau=\int_0^t \frac{1}{\lambda^2(s)}\, ds,
\end{equation*}
where $\lambda:[0,T)\rightarrow[0,\infty),\  T>0,$ is  a positive differentiable function (\textit{compression} or \textit{dilatation} parameter), such that $\lambda(t)\rightarrow 0$ as $t\uparrow T$. The advantage of passing to blowup variables is that the function $u$ is expected to have bounded derivatives and the blowup time is eliminated from consideration (it is mapped to $\infty$).  
Writing \eqref{m-eq} in blowup variables, we find the equation for the rescaled mass function
\begin{equation}\label{u-eq}
\p_\tau u =\Lap{0}_yu +y^{-1}u_\lam\p_y u  -a y\p_yu,
\end{equation}
where $a:=-\dot{\lambda}\lambda$. 
 Now,  the blowup problem for \eqref{m-eq} is mapped into the problem of asymptotic dynamics of solitons for  the equation \eqref{u-eq}, which was already studied in the pioneering works of \cite{SW1, SW2, SW3, BP, BS, TY1, TY2, TY3, GaSig1, GaSig2}.

 The boundary conditions on $u $ are $\p_y^\alpha u (y,\tau)\rightarrow 0$ as $y\rightarrow\infty$ for $\alpha=1,2$.  As with the boundary conditions for \eqref{m-eq}, these imply that mass is conserved: $\lim_{y\rightarrow \infty}u (y,\tau)=\lim_{y\rightarrow\infty}u (y,0)$.
Equivalently, $u $, as a solution of \eqref{u-eq}, depends on $a$, which determines $\lambda$, given $\lambda(0)=\lambda_0$, according to the formula
\begin{equation}\label{eqn:lambda.in.terms.of.a}
 \lambda^2(t)=\lambda_0^2-2\int_0^t a(s)\, ds.
\end{equation}

Equation \eqref{u-eq} has the static solution $(\chi(y), a=0)$.  
 It is shown in \cite{dlos} that the linearized operator on this solution  has one negative eigenvalue $-2a+\frac{a}{\ln\frac{1}{a}}+\O{a\ln^{-2}\frac{1}{a}}$
 (corresponding to the scaling mode---for a fixed parabolic scaling it is connected to possible variation of the blowup time) \footnote{A similar analysis applies also in the subcritical case $M<8\pi$ where the solution converges to a self-similar one as $\tau\ra\infty$,  which vanishes as $t\ra\infty$. In this case  the operator $\Lab$ has strictly positive spectrum.}  and one near zero 
eigenvalue, while the third eigenvalue, $2a+\frac{2a}{\ln\frac{1}{a}}+\O{a\ln^{-2}\frac{1}{a}}$, is positive, but vanishing as $a\ra 0$. (
 It also isolates the correct perturbation (adiabatic) parameter---$\frac{1}{\ln\frac{1}{a}}$.) Hence we have to construct a one-parameter deformation of $\chi(y)$ (besides the parameter $\lambda$, or $a$). For technical reasons it is convenient to use a two-parameter family,  $\chi_{bc}(y)$
\begin{equation}\label{chibc}
 \chi_{bc}(y):=\frac{4b y^2}{c+ y^2},
\end{equation}
with $b>1$ and both parameters $b$ and $c$ are close to $1$, with an extra relation between the parameters $a,\ b$ and $c$.
The  family  $\chi_{bc}(y)$ gives approximate solutions to \eqref{u-eq}  (see \eqref{chi-sol}) and forms  the deformation (or almost center-unstable) manifold $\cM:=\{\chi_{bc}(r/\lambda)\ |\ \lambda>0,\ p\}$. 
We expect that the solution to \eqref{u-eq}  approaches this manifold as $\tau\ra \infty$, and therefore
we decompose  the solution $u(y,\tau)$ to \eqref{u-eq} as the leading term, $\chi_{b(\tau)c(\tau)}(y)$, and the fluctuation, $\phi(y,\tau)$,
 \begin{equation} \label{orth-deco}
u(y,\tau)=\chi_{b(\tau)c(\tau)}(y)+\phi(y,\tau),
\end{equation}
and require that the fluctuation $\phi(y,\tau)$ is orthogonal to  the tangent space of  $\cM$ at $\chi_{b(\tau)c(\tau)}(y)$, 
$\lan\p_p\chi_{p(\tau)}(\cdot), \phi(\cdot,\tau)\ran=0,$ 
where $p:=(b, c)$. Note that this family evolves on a different spatial scale than $\phi(y, \tau)$ in \eqref{orth-deco}, as it can rewritten as $\chi_{bc}(y)=\chi_{\frac{b}{c}, 1}(\frac{y}{\sqrt{c}})=\chi_{bc}$.

 In parametrizing solutions as above, we split the dynamics of \eqref{KS3} into a finite-dimensional part describing motion over the manifold, $\cM$, and an infinite-dimensional fluctuation (the error between the solution and the manifold approximation) which is supposed to stay small.
  Substituting the decomposition \eqref{orth-deco} into the equation \eqref{u-eq}, we arrive at  the equation
\begin{equation}\label{phi-eq'}
 \p_\tau\phi=-\Labc\phi+\Fabc+\Nab(\phi),
\end{equation}
 where ${\cal L}$ is a self-adjoint linear operator, $\mathcal{F}$ is a forcing term, and ${\cal N}$ is a quadratic nonlinearity.
Due to the definition of  $\cM$, it turns out that its tangent space is very close to the subspace spanned by  the negative and almost zero
spectrum eigenfunctions (unstable modes) of the linearized operator, $\Labc$, and therefore $\phi$ is (approximately) orthogonal to the latter subspace. 

The contraction law is obtained by using  the orthogonality condition, $\lan\p_{bc}\chi_{bc}, \phi\ran=0$. The latter is equivalent to two conditions,
\begin{equation} \label{orthog}
\p_\tau\lan\p_{bc}\chi_{b(\tau)c(\tau)}(\cdot), \phi(\cdot,\tau)\ran=0
\end{equation}
and $\lan\p_{bc}\chi_{b(\tau)c(\tau)}(\cdot), \phi(\cdot,\tau)\ran|_{t=0}=0$, which  lead, to leading order, to the differential equation
\begin{align}\label{orth-eqns-lead7}
 a_\tau = -  \frac{2a^2}{\ln(\frac{1}{a})},
\end{align}
whose solutions, to leading order,  are \eqref{lambdaAs} (see Section \ref{sec:blowup-dyn}).

We now describe the organization of this paper.
In Section \ref{sec:param}, solutions to \eqref{KS3} are parametrized  by the parameters $(a, b, c, \phi)$ connected to $u$ by \eqref{u-eq}. In Section \ref{sec:Labc-genprop}, we study the operator $\Labc$  in \eqref{phi-eq'} and  show that it has one negative eigenvalue and one simple eigenvalue near zero.  We also give approximate eigenfunctions corresponding to these eigenvalues and prove that $\Labc$ is positive on the space orthogonal to these quasi-eigenfunctions.
In Section \ref{sec:blowup-dyn}, we state the  relationship between the blowup parameters $a,\ b$ and 
$c$, whose proof is given in Appendix \ref{sec:rel-abc}, and use it to 
obtain a dynamical equation for the blowup parameter $a=- \lambda\p_t\lambda$ and derive the leading order behaviour of the scaling parameter $\lambda$ in terms of the original time variable.
In Section \ref{sec:Labc-lowerbnd}, we derive the lower bounds on the operator $\Labc$. We use these bounds in
 in Section \ref{sec:fluct} in order to control the fluctuation $\phi$ in the linearized equation, i.e.\ for \eqref{phi-eq'}, with the nonlinearity $\Nab(\phi)$ omitted.

In Appendix \ref{sec:CompleteSetStaticSolns} we present the family of solutions to \eqref{m-eq},
\begin{equation*}
\chi^{(\mu)}(r):=\frac{r^{\mu-2}\mu+4-\mu}{r^{\mu-2}+1}
\end{equation*}
with mass $2\pi\mu$, where $\mu\in (2,4]$.  These solutions describe partial collapse with $2\pi(4-\mu)$ units of mass
concentrated at the origin. In the remainder of our work we will make no further use of these partially collapsed solutions. 
In Appendix \ref{sec:orthog-deco-pf} we provide a proof of the orthogonal splitting theorem of Section \ref{sec:param} and in  Appendix \ref{sec:grad} we discuss te gradient structure of equations \eqref{KS} and \eqref{KS3}.

In the following discussion, we use the notation $f\lesssim g$ if there exists a positive constant $C$ such that $f\le C g$ holds.  If the inequality $|f|\le C|g|$ holds then we write $f=\O{g}$.  We also write $f\ll g$ or $f=\smallO{g}$ if $f(a)/g(a)\rightarrow 0$ as $a\rightarrow 0$ and $f\sim g$ if the quotient converges to 1.

\medskip

\noindent {\bf Acknowledgements.} The research of the second and fourth authors is partially supported by NSERC under Grant NA7901, and of the third author, by NSF under Grants DMS 0719895 and DMS 0807131.

\section{Parametrization of Solutions}\label{sec:param}

We parameterize solutions $u_\lambda(y,\tau)$ of equation \eqref{u-eq} by the parameters $a,\ b$ and $c$, and the fluctuation $\phi$ according to
\begin{equation} \label{basic-deco}
u_\lambda(y,\tau)=\chi_{bc}(y)+\phi(y,\tau)
\end{equation}
where $\lambda =\lambda(\tau) ,\ b=b(\tau)$ and $a=a(\tau)$.
Substituting decomposition \eqref{basic-deco} into equation \eqref{u-eq} gives that the fluctuation $\phi$ satisfies
\begin{equation}\label{phi-eq}
 \p_\tau\phi=-\Labc\phi+\Fabc+\Nab(\phi),
\end{equation}
where 
the linear operator, the forcing terms, and the nonlinear term are
 \begin{align}
&\Labc :=-\Lap{4}-\frac{8bc}{(c+ y^2)^2}-\frac{4}{y}(b-1-\frac{bc}{c+ y^2})\p_y +a y\p_y,
\label{Labc}\\
&\Fabc :=-4b_\tau +4\frac{b_\tau c+bc_\tau-2bca}{c+ y^2}
+4bc\frac{8(b-1)+2ac-c_\tau}{(c+ y^2)^2}-\frac{32bc^2(b-1)}{(c+ y^2)^3},\label{Fabc}
\\ &\mathcal{N}(\phi):=y^{-1}\phi\p_y\phi .\label{Nabc}
\end{align}
Consider the weighted $L^2-$space $L^2(\Rp, \g_{a b}(y) y^3 dy)$, with the weight
\begin{equation}\label{gauge-bc}
\g_{a b c}^{-1/2}(y)=\frac{4 y^{2}e^{\frac{a}{4} y^2}}{(c+ y^2)^{b}},
\end{equation}
and the corresponding inner product
 \begin{equation}\label{ip-bc}
 \ip{f}{g}:=\int_0^\infty f(y) g(y)\, \g_{a b c}(y) y^3 dy. \end{equation}
 The norm corresponding to this inner will be denoted by $\|\cdot\|$. The significance of this space is 
that, as we show below, the operator $\Labc$ is self-adjoint on it. 

\medskip

\noindent {\bf Remark.} Another way to write $\Labc$ is as
 \begin{align}\label{Labc'}
&\Labc =-\Lap{0}-\frac{8bc}{(c+ y^2)^2}-\frac{4 by}{c+ y^2}\p_y +a y\p_y,
\end{align}
and treat it as a self-adjoint operator on $L^2(\Rp, \tilde\g_{a b}(y) y^3 dy)$, with weight
$\tilde\g_{a b c}^{-1/2}(y)=\frac{e^{\frac{a}{4} y^2}}{(c+ y^2)^{b}}, $ 
and corresponding inner product
$ \ip{f}{g}:=\int_0^\infty f(y) g(y)\, \tilde\g_{a b c}(y) y^{-1} dy.  $ 

\medskip

The decomposition \eqref{basic-deco} is not unique and as a result we have a single equation, \eqref{phi-eq}, for four unknowns, $a,\ b,\ c$ and $\phi$. Hence
we supplement equation \eqref{phi-eq} with three additional equations. Two of the equations can be chosen as in \cite{DGSW} to make the parameters  $a,\ b, $ and $c$ satisfy a chosen relation, say $f(a, b, c)=0$. In addition, we have the relations
\begin{equation} \label{orthoconditions-bc}
\ip{\phi}{\zic}=0,\ i=0,1,\
\end{equation}
in $L^2(\R^+, \g_{a b c}(y) y^3dy)$, for all times $\tau>0$,  where $\zic$ are the tangent vectors 
to the manifold $\cM:=\{\chi_{bc}(r/\lambda)\ |\ \lambda>0,\ b,\ c\}$:
\begin{align}\label{zbc}
\zeta_{bc0}(y):=\frac{1}{8 bc}y\p_y\chi_{bc} (y)&=\frac{ y^{2}}{(c+y^{2})^2},\  \quad \zeta_{bc1}(y):=\frac{1}{4}\p_b\chi_{bc} (y)=\frac{y^{2}}{c+ y^{2}},\\ &\zeta_{bc2}(y):=-\frac{1}{4b}\p_c\chi_{bc} (y)=\frac{y^{2}}{(c+ y^{2})^2}\notag.
\end{align}
(The vectors $\zeta_{bc0}(y)$ and $\zeta_{bc2}(y)$ are seen to be multiples of each other which confirms that one of the parameters is superfluous.) 

We proceed here differently and choose
\begin{equation}\label{orth-eqns}
-\ip{\phi}{\p_\tau\zic+(\p_\tau \ln \g_{a b})\zic}=-\ip{\Labc\phi}{\zic}+\ip{\Fabc}{\zic}+\ip{\Nab}{\zic},\ i=0, 1.
\end{equation}
 As we will show the latter vectors are approximate eigenvectors of  the operator $\Labc$ having the negative and almost zero eigenvalues. In addition, we will choose a relation between the parameters $a,\ b,\ c$.
Eqns \eqref{orth-eqns} imply that
\begin{equation} \label{orthog-concerv}\p_\tau\ip{\phi}{\zic}=0\,,
\end{equation} and therefore 
the inner products $\ip{\phi}{\zic},\ i=0, 1,$ are constant (one can think of this a constraint on $a, b$ and $c$).  The next proposition shows that $a_0, b_0$ and $c_0$ can be taken so that $\ip{\zic}{\phi}|_{\tau=0}=0$, and hence, by \eqref{orth-eqns}, we have \eqref{orthoconditions-bc}. To be able to formulate a precise statement we introduce, for a fixed $\delta>0$, open neighbourhoods of $\mathcal{M}$,
\begin{align}
\mathcal{U}_\eps=\{f:\norm{\e^{-\frac{\delta}{3}y^2}\lb f(y)-\chi_{bc}(y)\rb}_\infty<\eps\,, \textrm{for some } 1\leq b\leq 2, 
 \frac{1}{2}\leq c\leq 1 \}
\end{align}
and, for a fixed $\lambda>0$,
\begin{align}
\mathcal{\tilde{U}}_\eps=\{f(r): f(\frac{r}{\lambda})\in \mathcal{U}_\eps \}\,.
\end{align}
\begin{prop}\label{prop:orthog-deco} Fix $\lambda_0>0$ and $0<\delta\ll 1$. Then there is an $\eps>0$ and a unique $C^1$ function $g :\mathcal{U}_\eps \rightarrow (\delta,1)\times(1/2,1)$ such that for $m_0 \in \tilde{\mathcal{U}}_\eps$ we have the equation $\ip{\zeta_{b_0c_0i}}{\phi_0}=0$, or in detail, for $i=1,2$,
\begin{equation}\label{orthdeco-ic-bc}
\left.\ip{m_0(\lambda_0\cdot)-\chi_{b_0c_0}}{\z_{b_0c_0i}}\right|_{(a_0, c_0)=g(m_0)}=0.
\end{equation}
\end{prop}
For the proof of this proposition see Appendix \ref{sec:orthog-deco-pf}.

Equations \eqref{phi-eq} and \eqref{orth-eqns} form a system of coupled, partial and ordinary differential equations for the parameters $\phi$, $a$, $b$, and $c$.  We assume that this system has a unique local solution given initial conditions $\phi_0$, $a_0, b_0$ and $c_0$, the values of which are related to the initial value of $m$ (recall that
$u_\lambda( y)=m(\lambda y)$).

\section{General Properties of the Operator $\Labc$}
\label{sec:Labc-genprop}

 Before proceeding we discuss general properties of the operator $\Labc$ mentioned above and used below,
\begin{prop}\label{prop:Labc-sa}
The operator $\Labc$, defined on $ L^2([0,\infty), \g_{a b c}(y) y^3 dy)$ (with inner product \eqref{ip-bc}), is self-adjoint and has purely discrete spectrum.   Moreover, we have the lower bound
\begin{align}\label{cLabc-lowerbound}
\scalar{\phi}{\Labc\phi}\gtrsim -[2a+\frac{(b-1)+a}{\sqrt{\ln\frac{1}{a}}}]\norm{\phi}^2  \,.
\end{align}
\end{prop}
\begin{proof}  One can check the self-adjointness of $\Labc$ directly or use the unitary map $ \xi(y)\ra  \g_{a b c}^{1/2}(y)\xi(y)$, from $ L^2([0,\infty), \g_{a b c}(y) y^3 dy)$ to $ L^2([0,\infty), y^3 dy)$ to map this operator into the operator
\begin{align}\label{Lab-unit-cLab}
L_{a b c}:=\g_{a b c}^{1/2}\Labc\g_{a b c}^{-1/2},
\end{align}
acting on $L^2([0,\infty), y^3 dy)$ with inner product  $\ipp{\xi}{\eta}:=\int \xi\eta y^3 dy$. The latter operator can be explicitly computed to be
\begin{align}\label{Lab-unit}
L_{a b c} :=-\Lap{4}+\frac{1}{4}a^2 y^2-2ab+\frac{2b\lb 2(b-1)+ac\rb}{c+y^2}-\frac{4bc(b+1)}{(c+ y^2)^2}\,.
\end{align}
It is of Schr\"odinger type with the real continuous potential tending to $\infty$ as $y \rightarrow \infty$ as $O(y^2)$. Hence, using standard arguments (see e.g. \cite{GS}), one can show that $L_{a b c}$ is self-adjoint and its spectrum, and hence the spectrum of $\Labc$, is purely discrete.

Now, we investigate the bottom of the spectrum of the operator $\Labc$. We begin with the operator $\cL_{0bc}:=\Labc|_{a=0}$.
In what follows  we use the convenient shorthand notation $f_\lambda(r)=f(r/\lambda)$, which we apply only for the subscript $\lambda$.

\begin{lemma}\label{Lemma:SpectrumL0}
The function $\zeta_{bc}(y):=\frac{ y^2}{(c+ y^2)^{2b}}$ 
is a zero mode of the operator $\cL_{0bc}:=\Labc|_{a=0}$:
\begin{equation} \label{cLz}
\cL_{0bc}\zeta_{bc}=0.
\end{equation}
The spectrum of $\cL_{0bc}$ starts with $0$, which is a simple eigenvalue.
\end{lemma}
\begin{proof} Since $\chi_\lambda$ is a static solution to \eqref{m-eq}, differentiating the equation $\Lap{0}_r\chi_\lambda+r^{-1}\chi_\lambda\p_r\chi_\lambda=0$ with respect to $\lambda$ at $\lambda=1$, gives that
\begin{equation}\label{zeta}
\zeta:=r\p_r\chi=\frac{8r^{2}}{(1+r^{2})^2}
\end{equation}
 is a zero mode of the linearization of \eqref{m-eq} around $\chi$:
\begin{equation}\label{cL0}
\cL_0\zeta=0,
\end{equation}
where
$\cL_0:=-\Lap{0}_r+r^{-1}\p_r\cdot\chi$. (The vector $\zeta_\lambda$ spans the tangent space of the manifold ${\cal M}_0:=\{\chi_\lambda\ |\ \lambda>0\}$ at a point $\chi_\lambda$.)  We deform this result as \eqref{cLz}.
Consequently, by the Perron-Frobenius theorem, the spectrum of $\cL_{0bc}$ starts with $0$, which is a simple eigenvalue.
\end{proof}
 The results above can be translated to the operator
 $L_{bc}:=L_{a b c}|_{a=0}:=\g_{0 b c}^{1/2}\cL_{0 b c}\g_{0 b c}^{-1/2}$, which is explicitly given by
\begin{equation} \label{eqn:OpLb}
 L_{bc}=-\Lap{4}-\frac{8bc}{(c+ y^2)^2}
+\frac{4b(1-b) y^2}{(c+ y^2)^2}.
\end{equation}
This is a deformation in $b$ and $c$ of the operator 
\begin{equation} \label{L1}L_0:=-\Lap{4}_y-\frac{8}{( 1+y^2)^2}.\end{equation}
The operator  $L_{bc}$, defined on the Hilbert space $L^2([0,\infty),y^{3} dy)$, is self-adjoint with spectrum $[0,\infty)$.
The bottom of the spectrum, $0$, is a simple eigenvalue, 
 \begin{equation} \label{zeromode-bc-eqn}
  L_{bc} \eta_{bc} =0,\ \quad \mbox{with}\  \quad
 \eta_{bc}:=4\g_{0 b c}^{1/2}\zeta_{bc} =\frac{1}{(c+ y^2)^b}.
\end{equation}

The tangent vectors  $\zeta_{bc0}(y)$ and $\zeta_{bc1}(y)$ (to the manifold $\cM$) are approximate eigenfunctions of the operator $\Labc$. Indeed, first observe that the functions $\chi_{bc}(y)$ are approximate solution to \eqref{u-eq}. Indeed, let $\Phi(u)$ be the map defined by the right hand side of \eqref{u-eq}, $\Phi(u):=\Lap{0}_yu+y^{-1}u\p_yu -a y\p_y u$, then
\begin{align} \label{chi-sol}
\Phi(\chi_{bc})(y)=-\frac{8bca}{c+ y^2}
+8bc\frac{4(b-1)+ac}{(c+ y^2)^2}-\frac{32bc^2(b-1)}{(c+ y^2)^3}.\end{align}
Now, differentiating $\Phi(\chi_{bc})$ with respect to $c$ and $b$ and using \eqref{chi-sol}, we obtain
\begin{align}\label{Lzbc}
\Labc\zeta_{bc0}(y)&=-2a\zeta_{bc0}(y)+[4\frac{4(b-1)+ac}{c+ y^2}-24\frac{c(b-1)}{(c+ y^2)^2}
]\zeta_{bc0}(y),\notag\\
\Labc\zeta_{bc1}(y)&=[\frac{2ca}{c+ y^2}
-\frac{8c(2b-1)}{(c+ y^2)^2}]\zeta_{bc1}(y).
\end{align}
Though on the first sight $\zeta_{bc0}$ and, especially, $\zeta_{bc1}$ do not seem to be approximate eigenfunctions of $\Labc$, in fact they are. Indeed, assuming $b-1=O(a\ln\frac{1}{a}), c=O(1)$,
we obtain
\begin{align}\label{LZeta}
\|(\Labc+2a\delta_{i0})\zic\|&\lesssim \left\{
\begin{array}{ll}
(b-1)+a
 & i=0,\\
1 & i=1.
\end{array}
\right.
\end{align}
However, if one takes into account the normalizations
\begin{align}\label{zbc-norms}
\|\zeta_{bc0}\|=\frac{1}{4\sqrt{2}}\ln^{\frac{1}{2}}\frac{1}{a}+O(\frac{1}{\ln^{\frac{1}{2}}\frac{1}{a}}),\qquad \|\zeta_{bc1}(y)\|=\frac{\sqrt{2}}{a}+O(\ln^2\frac{1}{a}),
\end{align}
then, for the normalized vectors we have
\begin{align}\label{LZeta-norm}
\|(\Labc+2a\delta_{i0})\frac{\zic}{\|\zic\|}\|&\lesssim \left\{
\begin{array}{ll}
\frac{(b-1)+a}{\sqrt{\ln\frac{1}{a}}}
 & i=0,\\
a & i=1.
\end{array}
\right.
\end{align}
The relation for $i=0$ implies \eqref{cLabc-lowerbound}.  \end{proof}


\section{Relation between Parameters $a,\ b$ and $c$ and  Blowup Dynamics} \label{sec:blowup-dyn} 

In this section,  we state the relations between the parameters $a, b$ and $c$, which is obtained by evaluating the equations in \eqref{orth-eqns} and is proven in Appendix \ref{sec:rel-abc}. Using these relations, we find the governing equation for $a(\tau)$.
\begin{prop}\label{thm:RELATabc}
Let $d:=b-1$ and assume
\begin{align}\label{phi-assump}
\LpNorm{2}{\phi}\ll (\ln\frac{1}{a})^{-1}\,,
\end{align}
 and, for simplicity, 
 $d\ls a\ln(a^{-1})$. 
 Then
\begin{align}\label{abc-eq1}
&c_\tau +
S_{0}(\phi, a, b, c) a_\tau=4\lb\frac{ac}{2}-\frac{d}{\ln(\frac{1}{a})}\rb+ O(\frac{a}{\ln(\frac{1}{a})})+\mathcal{R}_0(\phi, a, b, c),\\
 &d_\tau+S_{1}(\phi, a, b, c)a_\tau =-\frac{2da}{\ln(\frac{1}{a})}+ O(\frac{a^2}{\ln(\frac{1}{a})}+a^2 d \ln(\frac{1}{a}))+\mathcal{R}_1(\phi, a, b, c),
\label{abc-eq2}\end{align}
where $S_{i}(\phi, a, b, c)$ and $\mathcal{R}_{i}(\phi, a, b, c),\ i=0,\ 1,$ satisfy the estimates
\begin{align}\label{Si}
&|S_{i}(\phi, a, b, c)|\lesssim \LpNorm{2}{\phi}(a\ln(a^{-1}))^{i-1},\\
&|\mathcal{R}_{i}(\phi, a, b, c)|\lesssim 
\frac{a^{i+1}}{\ln^{1-i}(a^{-1})}\LpNorm{2}{\phi}+\frac{a^{i}}{\ln(a^{-1})}\LpNorm{2}{\phi}^2. \label{Rem-i} \end{align}
\end{prop}

\begin{remark} We see from \eqref{Rem-i} that for the terms $ \mathcal{R}_i(\phi, a, b, c)$ in \eqref{abc-eq1} -  \eqref{abc-eq2} to be subleading we should have   $\LpNorm{2}{\phi}\ll (a\ln(a^{-1})^{-1})^{1/2}$.
\end{remark}
As was mentioned above, this proposition is proven in Appendix \ref{sec:rel-abc}.
 Now, we choose a relation between $a,\ c$ and $d$, so that the leading order term on the right hand side of the equation for $c_\tau$, \eqref{abc-eq1}, vanishes:
\begin{align}\label{ba-relat}
d=\frac{1}{2}a\ln(a^{-1})\,.
\end{align}
\begin{prop} \label{prop:aeqn}
Assume $\LpNorm{2}{\phi}\leq \sqrt{\frac{a}{\ln(a^{-1})}}$ and \eqref{ba-relat}. Then the function $ a(\tau)$ satisfies the differential equation
\begin{align}\label{a-eq}
 a_\tau = -  \frac{2a^2}{\ln(a^{-1})}\lb 1+O\lb\frac{1}{\ln(a^{-1})}\rb\rb,
\end{align}
which gives
\begin{equation}\label{a-asymp}
 a(\tau)=\dfrac{\ln \tau}{2\tau} \lb 1+ O(\frac{1}{\ln^2\tau})\rb.
\end{equation}
\end{prop}
\begin{proof} Plugging the relation \eqref{ba-relat} into \eqref{abc-eq2} and remembering that $d=b-1$, we obtain 
\begin{align}\label{a-eq1}
 &\frac{1}{2}  \big (\ln(a^{-1})- 1 
 +2 S_{2}(\phi, a, b, c)\big) a_\tau =-\frac{2da}{\ln(\frac{1}{a})}+ O(\frac{a^2}{\ln(\frac{1}{a})})+\mathcal{R}_2(\phi, a, b, c),
\end{align}
We see that to solve this equation for $a_\tau$,  we need $|S_2|\ll \ln (1/a) $, which in view of \eqref{Si} with $i=2$ requires that $\| \phi\| \ll  \ln (1/a)$.
Due to the conditions of the proposition and estimate \eqref{Rem-i} with $i=2$, the high order terms in equation \eqref{a-eq} give a small correction upon integration and the leading part can be integrated exactly yielding \eqref{a-asymp}. \end{proof}

\begin{remark}
The above expression for $a_\tau$ passes a consistency test: $a_\tau <0$ and $|a_\tau| \ll a^2$.  \end{remark}
\begin{prop}
 For $|T-t|\ll 1$, the scaling parameter $\lambda$, with $a=-\lambda\dot{\lambda}$ satisfying \eqref{a-asymp}, 
 is asymptotic to
\begin{equation}
 \lambda=k e^{-\frac{1}{4}\ln^2 \tau}
\label{BUlambda}
\end{equation}
where $\tau$ is related to $t$ by
\begin{equation}
k(T-t)=  \dfrac{\tau}{\ln \tau}e^{-\frac{1}{2}\ln^2 \tau}.
\label{BUt}
\end{equation}
\end{prop}
\begin{proof} Using the definition  $a=-\lambda\dot{\lambda}$ and the relation $\partial_t \lambda = \lambda^{-2}\partial_\tau \lambda $ we arrive at  $a=-\lambda^{-1}\partial_\tau \lambda$. Combining this with \eqref{a-asymp} we obtain the equation  $\lambda^{-1}\partial_\tau \lambda=-\dfrac{\ln \tau}{2\tau} \lb 1+ O(\frac{1}{\ln^2\tau})\rb$. Solving this differential equation gives \eqref{BUlambda}. Combining \eqref{BUlambda} with $\partial_t \tau = \lambda^{-2}$ gives a differential equation for $\tau(t)$ solving which in the leading order leads to \eqref{BUt}.
\end{proof}
Solving \eqref{BUt} for $\ln^2 \tau$  and substituting the result into \eqref{BUlambda} gives \eqref{lambdaAs}.

\section{Lower Bound for the Operator ${\cal L}_{abc}$} \label{sec:Labc-lowerbnd}
In this section we investigate  the linear operator $\Labc$, defined in \eqref{Labc}.
The main result of this section is the following lower bound on the quadratic form
$\langle \phi, {\cal L}_{abc}\phi\rangle,\ \phi\perp \zic$, where, recall, the vectors $\zic$ are defined in \eqref{zbc}.
\begin{thm} \label{thm:Lab-lowerbound}
For $|a| \ll 1,\ |b-1| \ll 1,\  |a_\tau|\ll a^2$, $|b_\tau|\ll 
a(1-b)$ and for any $\phi\in H^1([0,\infty), \g_{abc}(y) y^3 dy),$ $ \phi\bot \zic,\ i=0,1$ we have,  for some absolute constant $c>0$,
\begin{equation}\label{cLabc-lowerbound-H1}
\langle \phi, {\cal L}_{abc}\phi\rangle \ge c a \HsNorm{1}{\phi}^2. 
\end{equation}
\end{thm}
\begin{proof}
Recall that the operator $\Labc$ is unitarily equivalent to the operator $L_{abc}$,
\begin{align}\label{Lab}
L_{a bc }=\g_{a bc }^{1/2}\Labc\g_{a bc }^{-1/2}\,,
\end{align}
acting on the space  $  L^2([0,\infty), y^3 dy)$ with the inner product  $\ipp{\phi}{\eta}:=\int \phi\eta y^3 dy$. The latter operator has been explicitly computed in the proof of Proposition \ref{prop:Labc-sa} to be
\begin{equation} \label{Lab-deco}
L_{abc}=L_{*}+ W(y)- 2ab\,,
\end{equation}
where
\begin{equation} \label{defn:Lstar}
L_*:=-\Lap{4}_y-\frac{4bc(b+1)}{(c+y^2)^2}+\frac{4b(b-1)}{c+y^2}+\frac{1}{4} a^2 y^2
\end{equation}
and \begin{align*} W(y):= \frac{2bac}{c+y^2}\ge 0\,.\end{align*}
Since the lower bound of $L_{a bc}$ is equal to the lower bound of $\Labc$, we estimate the former lower bound. We observe that, like  $L_{abc}$, the operator $L_*$ is self-adjoint on $L^2([0,\infty),\, y^3 dy)$ and its spectrum is purely discrete, provided $a>0$. The latter property follows from the fact that the potential in \eqref{defn:Lstar} goes to infinity as $y \rightarrow \infty$. Moreover, $L_*\ge 0$. Indeed, write $L_*=L_{0bc}+\frac{1}{4}a^2 y^2$, where
\begin{equation*}
 L_{0bc}:=-\Lap{4}-\frac{4bc(b+1)}{(c+y^2)^2}+\frac{4b(b-1)}{c+y^2}\,.
\end{equation*}
Define $\eta_1(y):=\frac{1}{2\chi(y)}\zeta(y)$, where $\zeta$ is defined in \eqref{zeta}, so that $L_{011}\eta_1 =0$. We compute
\begin{equation} \label{defn:eta1}
\eta_1(y):=\frac{1}{2\chi(y)}\zeta(y)=\frac{1}{1+y^2} .
\end{equation}
An extension of relation \eqref{defn:eta1} leads to the equation
 \begin{equation} \label{eqn:zeromodebeqn}
  L_{0bc} \eta_{bc} =0,
 \end{equation}
 where  $\eta_{bc} $ is a deformation in $b$ and $c$ of  $\eta_1$ given by
\begin{equation} \eta_{bc}:=\frac{1}{(c+y^2)^{b}}.
 \label{eqn:zeromodeb}
\end{equation}
Since $\eta_{bc} >0$ we conclude, as in Lemma \ref{Lemma:SpectrumL0}, that $ L_{0bc} \ge 0$ (with the zero being a resonance of $ L_{bc}$).
This, together with $L_*=L_{0bc}+\frac{1}{4}a^2 y^2, $ implies that $L_*\ge 0$.
Next, we have
\begin{lemma} \label{lem:lowerboundLstar}
$\mbox{For}\ \phi\in H^2([0,\infty), y^3 dy),\  \phi\bot \zic,\ i=0,1,$ and $a\ll 1$, we have
\begin{equation} \label{lowerboundLstar}
\ipp{L_*\phi}{\phi}\ge \lb 
4 -O(\frac{1}{\sqrt{\ln\frac{1}{a}}}) \rb a  \|\phi\|_*^2,
\end{equation} 
where $\ipp{\phi}{\eta}:=\int \phi\eta y^3 dy$ is the inner product in $L^2([0,\infty), y^3 dy)$ and $\|\cdot\|_*$ is the corresponding norm.
\end{lemma}
\begin{proof}
To this end we will use the minimax principle for self-adjoint operators (see \cite{ReSiIV}), which states that the third eigenvalue, $\lambda_3$, of $L_*$
\begin{equation}\label{min-max}
 \lambda_3=\inf_{{\rm dim}\, V=3}\max_{\phi\in V}\frac{\ipp{L_*\phi}{\phi}}{\|\phi\|_*^2},
\end{equation}
where $V$ is an arbitrary subspace of $H^1(\R^4)$ and the estimate from \cite{dlos}  of $\lambda_3$:
\begin{equation}\label{exp:lambda3}
 \lambda_3=4a+\frac{Ca}{\ln\frac{1}{a}}(1+o(1))\,,
\end{equation}
for some constant $C$. Now, let $\eta$ be the minimizer to $\ip{L_*\phi}{\phi}$ over
\begin{align*} \{\phi\in H^2([0,\infty), y^3 dy)\ |\ \phi\bot\zic,\ i=0, 1,\ \|\phi\|_*=1\}\,.\end{align*}
Since $L_*$ is selfadjoint, $\eta$ can be chosen to be real.  Since the spectrum of $L_*$ is discrete, this minimizer exists.  By the linear independence of $ \zic,\ i=0, 1$ and orthogonality of $\eta$ to $ \zic,\ i=0, 1$, the three vectors $\eta$, $ \zic,\ i=0, 1$ span a three dimensional space.  The minimax principle then asserts that
\begin{equation} \label{est:speclam3}
 \lambda_3\le \max_{\phi\in W}\frac{\ipp{L_*\phi}{\phi}}{\|\phi\|_*^2},
\end{equation}
where $W:={\rm span}\{\eta,\, \zic,\ i=0, 1\}$. Let $\phi_{n a },\ n=0,1,$ be an appropriate orthonormal basis in ${\rm span}\{ \zic,\ i=0, 1\}$:
\begin{equation}
\phi_{na} :=\eta_{bc}\psi_{na},
\end{equation}
with 
\begin{equation} \label{def:psina}
\psi_{0 a}:=\sqrt{\frac{2}{\ln\frac{1}{a}}}\lsb 1+\O{\frac{1}{\ln\frac{1}{a}}} \rsb e^{-\frac{a}{4} y^2},\ \psi_{1a}:=\
 \lsb 1+\O{\frac{1}{\ln\frac{1}{a}}} \rsb(\frac{c_1}{\ln^\frac{1}{2}\frac{1}{a}}-c_2 a y^2)e^{-\frac{a}{4} y^2},
\end{equation}
for some positive constants $c_1$ and $c_2$.
 We write $\phi=\gamma_1\eta+\gamma_2\phi_{0a}+\gamma_3\phi_{1a}$ in the inner product $\ipp{L_*\phi}{\phi}$, where
\begin{equation}
|\gamma_1|^2+|\gamma_2|^2+|\gamma_3|^2=1,
\label{cond:gammas}
\end{equation}
and use self-adjointness of $L_*$ to obtain that
\begin{align}
\ipp{L_*\phi}{\phi}=&|\gamma_1|^2\ipp{L_*\eta}{\eta}+|\gamma_2|^2\ipp{L_*\phi_{0a}}{\phi_{0a}}+|\gamma_3|^2\ipp{L_*\phi_{1a}}{\phi_{1a}}\notag\\
& +2 \Re(\gamma_1\gamma_2^*)\ipp{\eta}{L_*\phi_{0a}}+
2 \Re(\gamma_1\gamma_3^*)\ip{\eta}{L_*\phi_{1a}} \notag\\
& +2 \Re(\gamma_2\gamma_3^*)\ipp{L_*\phi_{0a}}{\phi_{1a}} .\label{Lstarxixi}
\end{align}

We compute the various matrix elements on  the right hand side of equation \eqref{Lstarxixi}. Using that 
\begin{equation*}
\Lap{4}_y\phi_{n a}=\psi_{n a}\Lap{4}_y\eta_{bc}+\eta_{bc}\Lap{4}_y
\psi_{n a}+2\lb\p_y \psi_{n a}\rb\lb\p_y\eta_{bc}\rb,
\end{equation*}
we find
\begin{equation*}
L_*\phi_{n a}=\psi_{n a}L_{0bc}\eta_{bc}+\eta_{bc} H_a\psi_{n
a}-2\lb\p_y\eta_{bc}\rb\lb\p_y \psi_{n a}\rb. 
\end{equation*}
Using the facts that $L_{0bc}\eta_{bc}=0$, $H_a\psi_{0 a}=2a\psi_{0 a}$ and $H_a \psi_{1 a}=4a \psi_{1 a}+(8c_2a\sqrt{\frac{\ln\frac{1}{a}}{2}}-\sqrt{2}c_1a)\psi_{0 a}$ 
and computing $2\lb\p_y\eta_{bc}\rb\lb\p_y \psi_{n a}\rb$,
we obtain that
\begin{equation*}
\lb L_*-2 a n\rb\phi_{n a}= S_n, 
\end{equation*}
with
\begin{equation*}
S_0:= -\lb 2a(b-1)+\frac{2abc}{c+y^2}\rb\phi_{0 a}\,,
\end{equation*}
and
\begin{align*}
S_1:= &-2a\lb (b-1)-\frac{bc}{c+y^2}  \rb\phi_{1 a}\\
&-\sqrt{2}a\lb 4c_2\sqrt{\ln\frac{1}{a}}(b-1) -\frac{4c_2bc}{c+y^2}\sqrt{\ln\frac{1}{a}}+2c_1 \rb\phi_{0 a}\,.
\end{align*}
Using this and the fact that the functions $\phi_{i a}$ are normalized, we estimate
\begin{align} \label{est:InnerLPhi1}
& \ipp{L_*\phi_{0a}}{\phi_{0a}}=-\frac{a}{2}+\O{\frac{a}{\ln\frac{1}{a}}}\,,\\
& \ipp{L_*\phi_{0a}}{\phi_{1a}}=-\frac{c_1a}{\sqrt{2}\ln\frac{1}{a}}+\O{\frac{a}{\ln^\frac{3}{2}\frac{1}{a}}} \\
&\ipp{L_*\phi_{1a}}{\phi_{1a}}=\frac{a}{\ln^{\frac{1}{2}}\frac{1}{a}}(2+\frac{c_1}{\sqrt{2}})+\O{\frac{a}{\ln^{2}\frac{1}{a}}}.
\label{est:InnerLPhi2}
\end{align}

Let $P^\bot$ be the orthogonal projection onto the orthogonal complement of the two vectors $\phi_{0a}$ and $\phi_{1a}$.  We compute that
\begin{equation}\label{PpurpLstar2} \|P^\bot L_*\phi_{ia}\|_*=\O{\frac{a}{\ln^\frac{1}{2}\frac{1}{a}}},\ i=0, 1.  
\end{equation}
Using the estimates, \eqref{est:InnerLPhi1} - \eqref{PpurpLstar2}, together with \eqref{est:speclam3}  and  \eqref{Lstarxixi}, we find \begin{align}
\lambda_3\le \max_{\gamma_i} \big\{ |\gamma_1|^2\ipp{L_*\eta}{\eta}-\frac{a}{2}|\gamma_2|^2
+\O{\frac{a}{\ln^\frac{1}{2}\frac{1}{a}}}\big\}. \nonumber 
\end{align}
Now, since $L_*\geq 0$, we know that $\ipp{L_*\eta}{\eta}\ge 0$. Then the above relation implies
 \begin{align}
\lambda_3\le \ipp{ L_*\eta}{\eta}
+\O{\frac{a}{\ln^\frac{1}{2}\frac{1}{a}}}.
\end{align}
Using expression \eqref{exp:lambda3} for $\lambda_3$ in the last inequality, we obtain
that \begin{equation*}\ipp{ L_*\eta}{\eta}\ge
4\lb 1 -O(\frac{1}{\sqrt{\ln\frac{1}{a}}}) \rb a ,\end{equation*}
 which gives the inequality \eqref{lowerboundLstar}.
\end{proof}

Because of the decomposition \eqref{Lab-deco} and since $W(y) \ge 0$
and $0<b-1 \lesssim \frac{1}{\sqrt{\ln\frac{1}{a}}}$, we arrive at
\begin{equation}\label{Labc-lowerbound-L2}
( \phi, L_{abc}\phi) \ge \frac{3}{2} a \norm{\phi}^2_*\,,
\end{equation}
or, by the unitary map $\phi\mapsto\gamma_{abc}^{1/2}\phi$,
\begin{align}\label{cLabc-lowerbound-L2}
\ip{\phi}{\mathcal{L}_{abc}\phi} \ge \frac{3}{2} a \norm{\phi}^2\,.
\end{align}
To pass from this bound to \eqref{cLabc-lowerbound-H1}, we decompose $\lan \phi, \cL_{abc}\phi\ran=(1-\delta)\lan \phi, \cL_{abc}\phi\ran+\delta\lan \phi, \cL_{abc}\phi\ran$ and use \eqref{cLabc-lowerbound-L2} for the first term and $\cL_{abc}\ge -\Delta^{(4)} -C$, for some $C>0$, for the second one. Optimizing with respect to $\delta$ produces \eqref{cLabc-lowerbound-H1}.
\end{proof}

\section{Analysis of Fluctuations} \label{sec:fluct}

In this section, neglecting the nonlinearity $N(\phi)$, we find a bound $\LpNorm{2}{\phi}\lesssim |1-b|$ on the fluctuation $\phi$. 
Given that we expect, from \eqref{ba-relat}, that $1-b\sim \frac{a}{2}\ln\frac{1}{a}$, this is sufficient to close the estimates. 
Neglecting the nonlinearity $N(\phi)$ in \eqref{phi-eq}, we arrive at the linear equation
\begin{equation}\label{phi-eq-linear}
 \p_\tau \phi=-\Lab\phi+\Fab.
\end{equation}
More precisely, we have the following proposition,
\begin{prop}\label{prop:xiL2}
Assume $a,\ b$ and $\phi$ solve \eqref{phi-eq-linear} and \eqref{orth-eqns} (with $\Nab=0$),
which is equivalent to \eqref{orthoconditions-bc}, and are such that $1-b=\O{a\ln\frac{1}{a}}$ and $b_\tau=\O{a}$, and assume \eqref{a-asymp} holds. Then, for $\tau\gg 1$, $\phi$ satisfies the estimate
\begin{equation*}
 \LpNorm{2}{\phi}\lesssim \LpNorm{2}{\phi(0)}\lb\frac{2a}{\ln\frac{1}{a}}\rb^{\ln\frac{1}{a}} + a\ln\frac{1}{a}.
\end{equation*}
\end{prop}
\begin{proof}
We use a Lyapunov argument with Lyapunov functional $\phi\mapsto\LpNorm{2}{\phi}^2$.  The time derivative of this functional on solutions $\phi$ to \eqref{phi-eq-linear} is
\begin{equation}
 \p_\tau\LpNorm{2}{\phi}^2=-2\ip{\phi}{\mathcal{L}_{abc}\phi}+2\ip{\mathcal{F}_{abc}}{\phi}+\ip{\phi}{(\p_\tau \ln\g_{a b c})\phi} 
\label{timederalyap}
\end{equation}
We estimate right hand side of this relation. Let $\chi(y),\ \bar{\chi}(y)\geq 0$ be a smooth partition of unity, $\chi^2+\bar{\chi}^2=1$,  s.t. $\chi(y)$ is a cutoff function that equals 1 on the set $\{ay^2\leq\kappa\}$, for some convenient large constant $\kappa>0$, and is supported on $\{ay^2\leq 2\kappa\}$.  We have
\begin{proposition}For any $\phi\in H^1([0,\infty),\gamma_{abc}(y)y^3\d y), \phi\perp\zeta_{bci}, i=0,1$ we have, for some absolute constants $k_1,k_2>0$,
  \begin{align}\label{eq:low1}
    2\scalar{\phi}{\Labc\phi}\geq a\norm{\phi}_{L^2}^2+k_1a\norm{\phi}^2_{H^1}+k_2\scalar{\bar{\chi}\phi}{a^2y^2\bar{\chi}\phi}\,.
  \end{align}
\end{proposition}
\begin{proof}Since $\phi$ is orthogonal to the vectors $\zeta_{bci}, i=0,1$, and since $a$ and $b$ satisfy the conditions of Theorem \ref{thm:Lab-lowerbound} we have estimates \eqref{cLabc-lowerbound-H1} and \eqref{cLabc-lowerbound-L2}, which imply
  \begin{align}\label{eq:low2}
    2\scalar{\phi}{\Labc\phi}\geq \frac{3}{2}a\norm{\phi}^2_{L^2}+ka\norm{\phi}_{H^1}^2\,.
  \end{align}
Next, we estimate $\scalar{\phi}{\Labc\phi}$ in a different way. For the partition of unity defined after \eqref{timederalyap}, we have the IMS formula (see e.g. \cite{CFKS})
 \begin{align}\label{ims-form}\Labc=\chi\Labc\chi+\bar{\chi}\Labc\bar{\chi}-|\nabla\chi|^2-|\nabla\bar{\chi}|^2. \end{align}
   By \eqref{cLabc-lowerbound} we have $\chi\Labc\chi\gtrsim -a\chi^2$. Using the inner product $(\xi,\eta):=\int\xi\eta\d y$ and the notation $\bar{\phi}:=\bar{\chi}\phi$ we obtain
\begin{align*}
  \scalar{\bar{\phi}}{\Labc\bar{\phi}}=(\bar{\phi},\gamma_{abc}^{1/2}L_{abc}\gamma_{abc}^{1/2}\bar{\phi})\,,
\end{align*}
which, together with \eqref{Lab-unit} gives, for $\kappa$ large enough,
\begin{align}
\scalar{\bar{\phi}}{\Labc\bar{\phi}}&\geq (\bar{\phi},\gamma_{abc}^{1/2}[\frac{1}{4}a^2y^2-2ab-\frac{4bc(b+1)}{(c+y^2)^2}]\gamma_{abc}^{1/2}\bar{\phi})\\
&\geq \scalar{\bar{\phi}}{(\frac{1}{8}a^2y^2+\frac{1}{9}\kappa a)\bar{\phi}}\,.
\end{align}
Next, using that $|\nabla\chi|$ and $|\nabla\bar{\chi}|$ are of the form $\sqrt{\frac{a}{\kappa}}\tilde{\chi}$, where $\tilde{\chi}$ is supported between $ay^2=\kappa$ and $ay^2=2\kappa$, we compute $|\nabla\chi|^2+|\nabla\bar{\chi}|^2\simeq \frac{a}{\kappa}\tilde{\chi}$, which leads to
\begin{align}
\scalar{\phi}{(|\nabla\chi|^2+|\nabla\bar{\chi}|^2)\phi}\lesssim \frac{a}{\kappa}\norm{\tilde{\chi}\phi}^2_{L^2}\,.
\end{align}
Using the IMS formula \eqref{ims-form} and the estimates above we find
\begin{align}\label{eq:low3}
\scalar{\phi}{\Labc\phi}\geq -c a\norm{\chi\phi}_{L^2}^2+\frac{1}{9}\scalar{\bar{\chi}\phi}{(a^2y^2+\kappa a)\bar{\chi}\phi}-C \frac{a}{\kappa}\norm{\phi}^2_{L^2}\,,
\end{align}
for positive constants $c,\ C $. Now, write $\scalar{\phi}{\Labc\phi}=(1-\delta)\scalar{\phi}{\Labc\phi}+\delta\scalar{\phi}{\Labc\phi}$, and use \eqref{eq:low2} for the first term on the right hand side and \eqref{eq:low3} for the second one, and choose $\delta$ sufficiently small to arrive at \eqref{eq:low1}.
\end{proof}

We substitute expression \eqref{Fabc} for $\mathcal{F}_{abc}$ and observe that the orthogonality of $\phi$ to $\zic$, $i=0,1$, implies
\begin{align*}
 \ip{1}{\phi}= c \ip{\frac{1}{c+y^2}}{\phi}=c^2 \ip{\frac{1}{(c+y^2)^2}}{\phi}\,.
\end{align*}
to obtain
\begin{align} \label{Fxi}
 \ip{\mathcal{F}_{abc}}{\phi}=& 32bc(b-1)[\ip{\frac{1}{(c+y^2)^2}}{\phi}-\ip{\frac{1}{(c+y^2)^3}}{\phi}]\,.\nonumber
\end{align}
Using H\"older's inequality
in the above equality implies
\begin{equation}\label{Fxi-nbd}
 \ip{\mathcal{F}_{abc}}{\phi}=\O{(b-1)\ \LpNorm{2}{\lan y\ran^{4-\eps}\phi}}.
\end{equation}

Next, we estimate $\ip{\phi}{(\p_\tau \ln\g_{a b c})\phi}$. Using that $ \p_\tau \ln\g_{a b c}=-a_\tau y^2/2+2b_\tau\ln(c+y^2)+2bc_\tau/(c+y^2)$, we obtain
\begin{equation}\label{dloggxi-nbd''}
\ip{\phi}{(\p_\tau \ln\g_{a b c})\phi}= - \frac{1}{2} a_\tau\|y\phi\|_{L^2}^2+2b_\tau\|(\ln(c+y^2))^{1/2}\phi\|_{L^2}^2+2b c_\tau\norm{\frac{1}{\sqrt{c+y}}\phi}_{L^2}^2\,.
\end{equation}
By  \eqref{abc-eq2}
 and  \eqref{a-eq},  we have  $a_\tau <0$, $b_\tau <0$, and assuming $c<1$,   we have  by \eqref{abc-eq1}  and \eqref{ba-relat}, that  $c_\tau <0$. Hence
\begin{equation}\label{dloggxi-nbd'}
\ip{\phi}{(\p_\tau \ln\g_{a b c})\phi}\le - \frac{1}{2} a_\tau\|y\phi\|_{L^2}^2\,.
\end{equation}
Now, $y^2\leq \kappa/a$ on $\mathrm{supp }\chi$, which implies $\scalar{\chi\phi}{(\partial_\tau\ln\gamma_{abc})\chi\phi}\leq -\frac{a_\tau\kappa}{2a}\norm{\chi\phi}^2_{L^2}$. This, together with the relation $\scalar{\chi\phi}{(\partial_\tau\ln\gamma_{abc})\chi\phi}=\scalar{\chi\phi}{(\partial_\tau\ln\gamma_{abc})\chi\phi}+\scalar{\bar{\chi}\phi}{(\partial_\tau\ln\gamma_{abc})\bar{\chi}\phi}$, gives
\begin{align}\label{eq:8}
\scalar{\phi}{(\partial_\tau\ln\gamma_{abc})\phi}\leq -\frac{a_\tau\kappa}{2a}\norm{\chi\phi}^2_{L^2}-\frac{1}{2}a_\tau\norm{y\bar{\chi}\phi}^2_{L^2}\,.
\end{align}
Using the last estimate, together with \eqref{timederalyap}, \eqref{eq:low1}, \eqref{Fxi-nbd}, \eqref{eq:8}, \eqref{ba-relat} and \eqref{a-asymp}, we obtain, for some absolute constants $k_1,k_2,C>0$,
\begin{align*}
\partial_\tau\norm{\phi}_{L^2}^2\leq -a\norm{\phi}_{L^2}^2-k_1a\norm{\phi}_{H^1}^2-k_2\scalar{\bar{\chi}\phi}{(a^2y^2+\kappa a)\bar{\chi}\phi}+C(a\ln\frac{1}{a})\norm{\avg{y}^{-\frac{7}{4}+\eps}\phi}_{L^2}^2\,.
\end{align*}
Using $\p_\tau\LpNorm{2}{\phi}^2 =2 \LpNorm{2}{\phi}\p_\tau\LpNorm{2}{\phi}$, dropping the second and third terms (these terms can be used to control the nonlinearity) and dividing the resulting inequality by $\LpNorm{2}{\phi}$, we obtain
\begin{align}\label{diff-ineq}
 \p_\tau\LpNorm{2}{\phi}  \le -\frac{a}{2}\LpNorm{2}{\phi} 
 +Ca\ln\frac{1}{a}.
\end{align}
Now, integrating the last inequality gives that
\begin{equation} \LpNorm{2}{\phi} \lesssim e^{-\int_0^\tau a(s)\, ds}\LpNorm{2}{\phi(0)}
+\int_0^\tau e^{-\int_\sigma^\tau a(s)\, ds} \lb a\ln\frac{1}{a}\rb(\sigma)\, d\sigma.
\label{boundxiduhamel}
\end{equation}

We have computed that in the sense of asymptotic equivalence,
$ a(\tau)\sim\frac{\ln\tau}{2\tau}, $ 
 as $\tau\to\infty$ (see equation \eqref{a-asymp}). Consequently, as $\sigma\to\infty$, we compute that
\begin{equation*}
 \int_\sigma^\tau a(s)\, ds\sim\frac{1}{4}\ln (\tau\sigma)\ln\lb\frac{\tau}{\sigma}\rb
 =\ln\lb\frac{\tau}{\sigma}\rb^{\frac{1}{4}\ln(\tau\sigma)}.
\end{equation*}
Using
\begin{align*}
\e^{-\int_0^\tau a(s)}=\e^{-\int_0^{\sqrt{\tau}} a(s)}\e^{-\int_{\sqrt{\tau}}^\tau a(s)}\,,
\end{align*}
noting that the first term on the right hand side is uniformly bounded, and the second term is $\sim \tau^{-\frac{3}{16} \ln\tau}$, we obtain $\e^{-\int_0^\tau a(s)}=O(\tau^{-\frac{3}{16}\ln \tau})$. Using now
\begin{align*}
\tau=\frac{\ln\frac{1}{a}}{2a}\lb 1-O(\frac{1}{\ln\frac{1}{a}})\rb
\end{align*}
we obtain that the term involving the initial condition in \eqref{boundxiduhamel} is bounded as
\begin{equation}
e^{-\int_0^t a(s)\, ds}\LpNorm{2}{\phi(0)}\lesssim \LpNorm{2}{\phi(0)}\lb\frac{2a}{\ln\frac{1}{a}}\rb^{\ln\frac{1}{a}}\,.
\label{integralxi0bound}
\end{equation}
To bound the integral term in \eqref{boundxiduhamel} we begin by splitting the domain of integration into $[0,\alpha\tau]$ and $[\alpha\tau,\tau]$ for some $0<\alpha<1$ to be chosen later:
\begin{equation*}
 \int_0^{\alpha\tau} e^{-\int_\sigma^\tau a(s)\, ds}\lb a\ln\frac{1}{a}\rb(\sigma)\, d\sigma+\int_{\alpha\tau}^\tau e^{-\int_\sigma^\tau a(s)\, ds}\lb a\ln\frac{1}{a}\rb(\sigma)\, d\sigma.
\end{equation*}
Since $(a\ln\frac{1}{a})(\sigma)$ is decreasing and $\sigma\mapsto e^{-\int_\sigma^\tau a(s)\, ds}$ is increasing and both are positive, we can bound these terms from above by
\begin{equation*}
 e^{-\int_{\alpha\tau}^\tau a(s)\, ds}\int_0^{\alpha\tau} \lb a\ln\frac{1}{a}\rb(\sigma)\, \d\sigma+\lb a\ln\frac{1}{a}\rb(\alpha\tau)\int_{\alpha\tau}^\tau e^{-\int_\sigma^\tau a(s)\, ds}\d\sigma\,.
\end{equation*}
Since $e^{-\int_{\alpha\tau}^\tau a(s)\, ds}\sim C_\alpha\tau^{\frac{1}{2}\ln\alpha}$, the first term is bounded from above by $C_\alpha\lb a\ln\frac{1}{a}\rb(0)\tau^{1+\frac{1}{2}\ln\alpha}$. Taking $\alpha$ such that $\ln(\alpha)/2<-2$ the first term is $\lesssim \tau^{-1}$, and the second is bounded by $\lb a\ln\frac{1}{a}\rb(\alpha\tau)\sim C\ln^2\tau/\tau$. So we find
\begin{equation}
\int_0^\tau e^{-\int_\sigma^\tau a(s)\, ds}(a\ln\frac{1}{a})(\sigma)\, d\sigma\lesssim \lb a\ln\frac{1}{a}\rb(\tau).
\label{integralxibound}
\end{equation}
Using bounds \eqref{integralxi0bound} and \eqref{integralxibound} in \eqref{boundxiduhamel} 
 completes the proof.
 \end{proof}

\appendix
\section{Complete Set of Static Solutions for the Radial rKS}
\label{sec:CompleteSetStaticSolns} The static solutions of equation
\eqref{m-eq} satisfy the second order differential equation
\begin{equation}
\p_r^2 \chi+\frac{1}{r}(\chi-1)\p_r\chi=0 \label{eqn:mStatic}
\end{equation}
and hence form a two dimensional manifold. We prove
\begin{prop}
Equation \eqref{eqn:mStatic} has the one-parameter family of
static solutions
\begin{equation*}
\chi^{(\mu)}(r):=\frac{r^{\mu-2}
\mu+ 4-\mu}{1+r^{\mu-2}},\  \mu\in[2,\infty),
\end{equation*}
(and therefore  the two-parameter family $\chi^{(\mu)}_{\lambda}(r):=\chi^{(\mu)}(r/\lam)$ as well). The mass at infinity of $\chi^{(\mu)}_{\lambda}$ is $\mu$.
\end{prop}
\begin{remark}
If $\mu < 4$, then the mass at the origin is non-zero, i.e. blowup
has already occurred.  If $\mu>4$, then the mass at the origin is
negative and hence the static solution is not physical.
\end{remark}

\begin{proof}
We use the transformation
\begin{equation*}
\psi(\chi)=r\frac{\p_r\chi}{\chi}
\end{equation*}
in \eqref{eqn:mStatic} under the assumption that the right hand side
is indeed a function of $\chi$ alone.  Using this transformation,
equation \eqref{eqn:mStatic} becomes
\begin{equation*}
\chi\p_\chi\psi+\psi=2 -\chi.
\end{equation*}
Integrating this equation gives that
\begin{equation*}
\psi=2-\frac{1}{2}\chi+\frac{\mu}{2}\frac{1}{\chi}
\end{equation*}
and hence, upon substituting this into the definition of $\psi$ and
integrating over $r$, we obtain the general solution
\begin{equation*}
\chi=\frac{\lb\frac{r}{\lambda}\rb^{\sqrt{4+\nu}} r_+ +
r_-}{1+\lb\frac{r}{\lambda}\rb^{\sqrt{4+\nu}}},
\end{equation*}
where $r_\pm=2\pm\sqrt{4+\nu}$ are the roots of
$\chi^2-4\chi-\nu=0$. The total mass at infinity of these solutions
is $r_+$ and hence it is natural to define a new parameter
$\mu=r_+\in[2,\infty)$.  The static solution in terms of the
parameters $\lambda$ and $\nu$ are
\begin{equation*}
\chi=\frac{\lb\frac{r}{\lambda}\rb^{\mu-2}\mu +
4-\mu}{1+\lb\frac{r}{\lambda}\rb^{\mu-2}}
\end{equation*}
The constant $\lambda$ is positive since it is the exponential of
the constant obtained in the last integration.
\end{proof}
The tangent space of the manifold $M_{\lambda,\mu}$ is spanned by
the functions
\begin{equation*}
\zeta^0_{\lambda,\mu}:=\p_\lambda\chi_{\lambda,\mu}=-\frac{2(\mu-2)^2}{\lambda}\frac{y^{\mu-2}}{(1+y^{\mu-2})^2}
\end{equation*}
and
\begin{equation*}
\zeta^1_{\lambda,\mu}:=\p_\mu\chi_{\lambda,\mu}=\frac{y^{2(\mu-2)}-1-2
y^{\mu-2}\ln y}{(1+y^{\mu-2})^2},
\end{equation*}
where $y=\frac{r}{\lambda}$.

We again restrict to the situation of $\kappa=1$ and $n=4$.  After
the gauge transform the zero modes $\zeta^0_{\lambda,\mu}$ and
$\zeta^1_{\lambda,\mu}$ transform to
\begin{equation*}
\eta^\lambda_{\lambda,\mu}:=\frac{1}{\lambda^2}\frac{1+y^2}{y^2}\zeta^0_{\lambda,\mu}\
\mbox{and}\
\eta^\mu_{\lambda,\mu}:=\frac{1}{\lambda^2}\frac{1+y^2}{y^2}\zeta^1_{\lambda,\mu},
\end{equation*}
neither of which are in $L^2(r^3\, dr)$ and hence are generalized
eigenfunctions of ${\cal L}$ (without the $\dot{\lambda}$ term. By
the ODE theory the above functions are the only linearly independent
solutions to the equation ${\mathcal L}_0\phi=0$.  The
Perron-Frobenius theory shows that 0 is the lowest point of the
spectrum of ${\cal L}$.

\section{Proof of Proposition \ref{prop:orthog-deco}}
\label{sec:orthog-deco-pf}
Both existence and uniqueness follow from a standard implicit function theorem argument. Fix $0<\delta\ll 1$  and let $Z:=e^{\frac{\delta}{3} y^2}L^\infty([0, \infty))$. Recall that $b=1+1/2 a \log(1/a)$, and define the vector-valued function
\begin{align}
 G(f,a,c)&:=\lb\ip{f-\chi_{bc}}{\zeta_{bci}},\ i=0,1\rb\\
&=\lb\frac{1}{16}\int_0^\infty \lb f(y)-\chi_{bc}(y) \rb (c+y^2)^{2b-2+i}y e^{-\frac{a}{2} y^2}\, \d y,\ i=0,1\rb\,.
\end{align}
This function maps $Z\times\R_+\times\R_+$ into $\R^2$.  It is a $C^1$ function and $G(\chi_{bc},a,c)=0$. Moreover, the derivative of $G$ with respect to $(a,c)$ at $f=\chi_{bc}$ is
\begin{equation}
A:= \lb
\begin{array}{cc}
 \Gamma_{0a} & \Gamma_{0c}\\
\Gamma_{1a} & \Gamma_{1c}
\end{array}
\rb,
\end{equation}
where
\begin{align}
 \Gamma_{ia}&:=\frac{\partial_ab}{4}\int_0^\infty y^3(c+y^2)^{2b-3+i}\e^{-\frac{a}{2}y^2}\d y\\
 \Gamma_{ic}&:=-\frac{b}{4}\int_0^\infty y^3(c+y^2)^{2b-4+i}\e^{-\frac{a}{2}y^2}\d y\,.
\end{align}
Compute that the determinant of $A$ satisfies
\begin{align*}
|\det A|=\frac{1}{64a^2}(\ln^2\frac{1}{a}+O(1)),
\end{align*}
as $a\to 0$, and so $|\det A|\geq C>0$, for some constant $C$, for $(a,c)\in (0,\delta)\times (1,2)$ for $\delta$ small enough. Thus by the implicit function theorem, for any $a_*\in(0,\delta)$ and $c_* \in(1,2)$ there exist open sets $U_{a_*c_*}\subset Z$ and $V_{a_*c_*} \subset (0,\delta)\times(1,2)$ containing $\chi_{b_*c_*}$ and $(a_*,c_*)$, respectively, and a unique function $g_{a_*c_*}:U_{a_*c_*}\rightarrow V_{a_*c_*}$. To determine the size of the neighbourhoods $U_{a_*c_*}$ we look more closely into a proof of the implicit function theorem. Write $\mu=(a,c)$ and expand
\begin{align}\label{eq:expand}
G(f,\mu)=G(f,\mu_*)+\partial_\mu G(f,\mu_*)(\mu-\mu_*)+R_f(\mu)\,,
\end{align}
where $R_f(\mu)=O(|\mu-\mu_*|^2)$ uniformly in $f\in B_C(\chi_{b_*c_*})$ and $(a_*,c_*)\in(\delta/2,\delta)\times (1,2)$, for any fixed constant $C$. By continuity and the above computations, there is $\eps>0$ such that $\det \partial_\mu G(f,\mu_*)$ is bounded away from zero uniformly for $f\in B_{\eps}(\chi_{b_*c_*})$ and $(a_*,c_*)\in(\delta/2,\delta)\times (1,2)$. From \eqref{eq:expand} we find a fixed-point equation for $\mu-\mu_*$,
\begin{align*}
  \mu-\mu_*=\Phi_f(\mu-\mu_*)\,,
\end{align*}
where
\begin{align*}
\Phi_f(\mu)=-(\partial_\mu G(f,\mu_*))^{-1}(G(f,\mu_*)+R_f(\mu))  \,.
\end{align*}
The above observations imply that there is an $\eps_1>0$ such that $\Phi_f$ is a contraction on $B_{\eps_1}(\mu_*)$ for any $f\in B_{\eps}(\chi_{b_*c_*})=:U_{a_*c_*}$. Taking the union of $U_{ac}$ over $a\in(\delta,1)$ and $c\in(1/2,1)$ gives the open set $\mathcal{U}_\eps$.  Patching together the functions $g_{ac}$ gives $g$. $\qed$

\section{Gradient Formulation} \label{sec:grad}


The Keller-Segel models \eqref{KS} and \eqref{KS3} are gradient systems.  We begin by formulating a normalized version of \eqref{KS},
\begin{equation}
\begin{split}
\p_t\rho&=\Delta\rho-\nabla\cdot(f(\rho)\nabla c)\\
\varepsilon\p_t c & =\Delta c+\rho-\g c,
\end{split}
\label{eqn:KS1norm}
\end{equation}
as a gradient system.  This system is obtained from \eqref{KS} by setting unimportant constants to 1.

Define the energy (or Lyapunov) functional
\begin{equation}
\Ef(\rho,c):=\int_\Omega \frac{1}{2}|\nabla c|^2-\rho c+\frac{\g}{2}
c^2+G(\rho)\, dx,
\label{eqn:Ef}
\end{equation}
where $G(\rho):=\int^\rho g(s)\, ds$ and $g(\rho):= \int^\rho
\frac{1}{f(s)}\, ds$. The $L^2$-gradient of $\Ef(\rho,c)$ is
\begin{equation*}
\grad_{L^2}\Ef(\rho,c)=\lb\begin{array}{c}-c+g(\rho)\, \\-\Delta
c-\rho+\g c\end{array}\rb,
\end{equation*}
and hence, if we define $U=(\rho,c)$, then \eqref{eqn:KS1norm} can be written in the form $\p_t U=I\Ef'(U)$, where
\begin{equation*}
I=\lb\begin{array}{cc}\nabla\cdot f(\rho)\nabla & 0\\0
&-\frac{1}{\varepsilon}\end{array}\rb.
\end{equation*}
The operator $I$ is non-positive and may be degenerate, however, assuming it is invertible, the operator $I$ defines the metric $\ip{v}{w}_I:=-\ip{v}{I^{-1}w}_{L^2\oplus L^2}$.  In this metric, $\mathrm{grad}\,\Er(U)=-I \Er'(U)$ and hence
\begin{equation*}
\p_t U=-\mathrm{grad}\,\Ef(U).
\end{equation*}
This shows that \eqref{eqn:KS1norm} has the structure of a gradient system.
A consequence of this is that the energy decreases on solutions of the KS system.  Indeed, if $f>0$, then
\begin{equation*}
\p_t\Ef(\rho,c)=-\LpNorm{2}{f(\rho)^\frac{1}{2}\nabla\lb
c-g(\rho)\, \rb}^2-\frac{1}{\varepsilon}\LpNorm{2}{\Delta
c+\rho-c}^2.
\end{equation*}

The gradient formulation for \eqref{KS3} is similar to the one for \eqref{eqn:KS1norm}. Instead of \eqref{eqn:Ef}, one uses the energy \eqref{energy}.  
The latter is obtained from \eqref{eqn:Ef} by dropping the quadratic term $\frac{1}{2}c^2$, 
replacing $c$ with $-\Delta^{-1} \rho$ in the remaining terms and using that $f(\rho)=\rho$.  
 The formal G\^{a}teaux derivative of $\Er$ is $\p_\rho\Er(\rho)\phi=\int (\Delta^{-1}\, \rho \,+ \ln\rho) \phi,$
and therefore the gradient in the metric   $\ip{v}{w}_J:=-\ip{v}{J^{-1} w}_{L^2}$, where $J:=\nabla\cdot\rho\nabla<0$, is
\begin{equation*}
\grad \Er(\rho)= -\nabla\cdot\rho\nabla (\Delta^{-1}\, \rho \,+ \ln\rho)=-\nabla\cdot\rho\nabla \Delta^{-1}\, \rho \,- \Delta\rho,\end{equation*}
which is the negative of the r.h.s. of  the first equation in \eqref{KS3} with  $c=-\Delta^{-1} \rho$.
Hence the equation \eqref{KS3} can be written as 
 $\p_t\rho=-\mathrm{grad}\, \Er(\rho)$ in the space with metric
$\ip{v}{w}_J:=-\ip{v}{J^{-1} w}_{L^2}$.  Again, the
energy $\Er$ decreases on solutions of \eqref{KS3}:
\begin{equation} \label{ptE}
\p_t\Er=\ip{\Er'}{I\Er'}=-\LpNorm{2}{\rho^\frac{1}{2}\nabla\Er'}^2.
\end{equation}
This can be thought of as an entropy monotonicity formula.

The stationary solutions of \eqref{KS3} are critical points of the energy functional ${\cal E}$, given in \eqref{energy}, under the constraint that $\int\rho=const.$. Thus, they  satisfy ${\cal E}'(\rho)=C$,
where $C$ is a constant.  
Explicitly ${\cal E}'(\rho)=C$ reads
\begin{align} \label{stateqn}
\log(\rho)+\frac{1}{\Delta}\rho=C\ \Leftrightarrow\
\Delta\log(\rho)+\rho=0\ \Leftrightarrow\ \Delta u+e^u=0,
\end{align}
where $u=\log(\rho)$.  Solutions to \eqref{stateqn} can be
written in the form of 'Gibbs states' $\rho=M \frac{e^c}{\int e^c}$ (see \cite{GaZa1998}),
with the concentration $c$ considered as a negative potential (remember that $\Delta c = -\rho$).  In two dimensions, this equation has the solution
$R=\frac{8}{(1+|x|^2)^2}.$
This solution is a minimizer of ${\cal E}$ under the constraint that $\int\rho=8\pi.$

\section{Proof of Proposition \ref{thm:RELATabc}}
\label{sec:rel-abc}
In this appendix, we prove Proposition \ref{thm:RELATabc}, relating  the parameters $a, b$ and $c$, by evaluating the equations in \eqref{orth-eqns}.
\begin{proof}[Proof of Proposition \ref{thm:RELATabc}] Let $R_{i}(\phi):=
\ip{{\Labc}\phi}{\zic}-\ip{{\cal N}}{\zic}.$ Here and in what follows $i=0,1$.
The equations \eqref{orth-eqns} can be rewritten as
\begin{equation}\label{orth-eqns-0}
\ip{\Fabc}{\zeta_{bci}}+\ip{\phi}{\p_\tau\zic+(\p_\tau \ln \g_{a b})\zic}=R_i(\phi).
\end{equation}
We
begin with evaluating $\ip{\Fabc}{\zeta_{bci}}$ to \textit{leading order}. To this end, we begin with the elementary computation 
\begin{align}\label{de-est1} 
&\ip{1}{\zeta_{bci}}=2^{i-4}a^{-i-1}+
O(\frac{1}{a^i}\ln^{2}\frac{1}{a}),  \\
&\ip{\frac{1}{c+ y^2}}{\zeta_{bci}}= 2^{i-5}a^{-i}\ln^{1-i}\frac{1}{a}+
O(\ln^{2i}\frac{1}{a}),  \label{de-est2}\\
&\ip{\frac{1}{(c+ y^2)^2}}{\zeta_{bci}}= 2^{-5}c^{i-1}\ln^i\frac{1}{a}+O(a^{1-i}\ln^{1-i}\frac{1}{a}), \label{de-est3}\\
& \ip{\frac{1}{(c+ y^2)^3}}{\zeta_{bci}}=2^{i-6}c^{i-2}+
 O(a\ln\frac{1}{a}).\label{de-est4}\end{align}
These estimates are proven at the end of this appendix. Using these estimates in \eqref{Fabc}, we arrive at
\begin{align}\label{eq:fabc}&\ip{\Fabc}{\zeta_{bci}}=- b_\tau 2^{i-2} a^{-i-1}\lb 1-c 2^{-1} a\ln^{1-i}\frac{1}{a}\rb +c_\tau\lb b 2^{i-3} a^{-i}\ln^{1-i}\frac{1}{a}-bc^i2^{-3}\ln^i\frac{1}{a}  \rb\nonumber\\
& -2^{i-2} bc a^{1-i}\ln^{1-i}\frac{1}{a}+(bc^id+abc^{i+1}2^{-2})\ln^i\frac{1}{a}-bc^id2^{i-1}\nonumber\\&
 +O(b_\tau a^{-i}\ln^{2}\frac{1}{a})+O((a_\tau+a)\ln^{2i}\frac{1}{a})+O((d+a+c_\tau) a^{1-i}\ln^{1-i}\frac{1}{a})+ O(d a\ln\frac{1}{a})\,.
\end{align}
Next, we compute the term $\ip{\phi}{\p_\tau\zic+(\p_\tau \ln \g_{a b c})\zic}$. Differentiating $\zic$ and $\p_\tau \ln \g_{a b}$ with respect to $\tau$, we obtain
\begin{equation} \label{dtauzita}
 \p_\tau\zic+(\p_\tau \ln \g_{a b c})\zic=
 \lsb 2b_\tau \ln(c+ y^2)+(2b-2^{1-i})\frac{c_\tau}{c+ y^2}-\frac{a_\tau}{2} y^2 \rsb\zic.
\end{equation}
Using in the case $i=1$ that $\phi$ is orthogonal to $\zeta_{bc0}$, we find
 \begin{align}\label{ip.phi.zdt}
\ip{\phi}{\p_\tau\zic+(\p_\tau \ln \g_{a b})\zic}=b_\tau S_{i1}(\phi)+c_\tau S_{i2}(\phi)-a_\tau S_{i3}(\phi), 
\end{align}
where
\begin{align*}S_{i1}(\phi)&:= 2\ip{\phi}{  \ln(c+ y^2)  \zic}\\
S_{i2}(\phi)&:=2\delta_{i0}( b-1)\ip{\phi}{\frac{1}{c+ y^2}\zic}\\
S_{i3}(\phi)&:= \frac{1}{2}\ip{\phi}{ y^2\zic}.
\end{align*}
Collecting \eqref{eq:fabc} and \eqref{ip.phi.zdt}, we have, for $i=0$,
\begin{align}\label{eq:i0}
&(  S_{03}(\phi) +O(1 ))a_\tau +  [ \frac{1}{4a}-\frac{c}{8}\ln(a^{-1})-S_{01}(\phi)+O(\ln^{2}\frac{1}{a})]b_\tau\notag\\
&-[\frac{b}{8}(\ln(a^{-1})-1)+S_{02}(\phi)+O( a \ln\frac{1}{a})]c_\tau 
=-\frac{b}{2}[\frac{1}{2}ac\ln(a^{-1})-d-a c]-R_0(\phi)\,.
\end{align}
and, for $i=1$,
\begin{align}\label{eq:i1}
 (S_{13}(\phi)a+O(\ln^{2}\frac{1}{a}) )a_\tau&+[ \frac{1}{2a}- \frac{c}{4}-aS_{11}(\phi)+O(a^{-1}\ln^{2}\frac{1}{a})]b_\tau-[\frac{b}{4}-\frac{1}{8}abc\ln(a^{-1})+aS_{12}(\phi)+O( 1)]]c_\tau\nonumber\\
&=-abc[\frac{1}{2}- d\ln(a^{-1})- \frac{ac}{4}\ln(a^{-1})+d]-aR_1(\phi)\,.
\end{align}

 We manipulate  equations \eqref{eq:i0} and \eqref{eq:i1} and solve them for $b_\tau$ and $c_\tau$  to obtain
 \begin{align}\label{atauctau-eq}
 &- f_0 a_\tau+gc_\tau =\frac{1}{8}b c\ln(a^{-1})-\frac{b d}{4a}  + v_0+r_0,\  \quad - f_1 a_\tau+ gb_\tau =-\frac{b^2 d}{8} + v_1+r_1,\end{align}
where
\begin{align} 
&f_0:= (\frac{1}{2a}-\frac{c}{4}-a S_{11})S_{03} +(\frac{1}{4}+\frac{c}{8}a\ln(a^{-1})-aS_{01})S_{13},\nonumber\\
& g :=\frac{b}{16}\frac{\ln(a^{-1})}{a} - \frac{1}{4} \big(\frac{b}{2a}-\frac{b c}{8}-\frac{1}{16} b c^2 a\ln(a^{-1})^2 -\frac{1}{8} b c \ln(a^{-1})-(\frac{2}{a}-c)S_{02}-b S_{01} + S_{12}\big)\nonumber\\
&\qquad -\frac{a}{8}  \ln(a^{-1})\big(b c S_{01}+ b S_{11}- c S_{12} \big)+a \big(\frac{1}{8}b S_{11}-S_{02} S_{11}+S_{01} S_{12}\big),   \nonumber\\
&v_0:=  \frac{1}{8}b c [2d\ln(a^{-1})+3 +d+ c  a\ln(a^{-1})(d  \ln(a^{-1}) - \frac{1 }{2}-d)- c a-\frac{1}{4} a^2  c^2 \ln(a^{-1})^2], \nonumber
\end{align}
\begin{align}
&r_0 := (\frac{1}{2a}-\frac{1}{4}c) R_0-\frac{1}{4}R_1 +  c a\ln(a^{-1}) (  b  d S_{01}-\frac{1}{8} R_1  )\nonumber\\
&\qquad  + a \big( \frac{1}{2} b c S_{01} +\frac{1}{2}b d S_{11}+R_0 S_{11}+ S_{01}R_1\big)  +\frac{1}{2}a^2 b c (  S_{11}-\frac{1}{2}\ln(a^{-1})  (c S_{01}+ S_{11} )), \nonumber\\
&f_1 :=a(\frac{b}{8}\ln(a^{-1})-\frac{b}{8}+S_{02}) S_{13}-(\frac{b}{4}-\frac{bc}{8}a\ln(a^{-1})+aS_{12})S_{03},\nonumber
\end{align}
\begin{align}
&v_1:=- \frac{1}{8} b^2 c[  d a\ln^2(a^{-1}) - \frac{3}{2} d a\ln(a^{-1}) - a (\frac{1}{2}-  d)+ \frac{1}{4}   c  a^2  \ln(a^{-1})],\nonumber\end{align}
\begin{align} & r_1:=-\frac{b}{4} R_0+b (\frac{ 1}{8} c R_0+-  d c S_{02}+\frac{1}{8}  R_1 )a\ln(a^{-1}) , \nonumber\\
&\qquad + a \big(\frac{bc S_{02}}{2}+b c d S_{02}+\frac{b d S_{12}}{2}-R_0 S_{12}-\frac{1}{8}b R_1+ R_1 S_{02}\big),\nonumber\\
&\qquad +\frac{1}{2} b ca^2 \big( S_{12}+\frac{1}{2} c S_{02} \ln(a^{-1})+\frac{1}{2}  S_{12} \ln(a^{-1})\big).\nonumber
\end{align}

Next, we derive estimates on $S_{i1}(\phi),\ S_{i2}(\phi),\ R_0(\phi)$ and $R_1(\phi)$.  Using the Cauchy-Schwarz inequality and simple modifications of the estimates \eqref{zbc-norms}, we arrive at the estimates
 \begin{align}\label{Si1-est}
|S_{i1}(\phi)|&\lesssim\LpNorm{2}{\phi}a^{-i}\ln(a^{-1})^{(3-i)/2},\  |S_{i2}(\phi)|\lesssim
\|\phi\|\ln(a^{-1})^{\frac{i}{2}}\,,\ 
 |S_{i3}(\phi)|\lesssim\LpNorm{2}{\phi}a^{-(i+1)}.
\end{align}

As was shown above, the operator $\Labc$ is self-adjoint in the inner product \eqref{ip-bc} and hence $\ip{{\Labc}\phi}{\zic}=\ip{\phi}{{\Labc}\zic}$. 
  Using \eqref{LZeta} and the fact that $\z_{bc0}$ is orthogonal to $\phi$, we obtain the estimate
\begin{align}\label{LXiZeta-est}
&|\ip{\Labc\phi}{\zic}|\lesssim \LpNorm{2}{\phi}(d+a)^{1-i}\,.
\end{align}
Lastly, we estimate $\ip{\Nab}{\zic}$ which can be written, using integration by parts, in the form
\begin{equation*}
 \ip{\Nab}{\zic}=-\frac{1}{2}\int_0^\infty \phi^2\p_y(\g_{abc}y^2\zic)\, dy,
\end{equation*}
where,  recall, $\g_{abc}$ is the gauge function (see \eqref{gauge-bc}).
Here we used that $y\phi\rightarrow 0$ as $y\rightarrow\infty$ so that the boundary terms vanish.
 Using that
 $$\p_y(\g_{abc}y^2\zic)=\p_y(\frac{(c+ y^2)^{2b}}{16 y^2}e^{-\frac{a}{2} y^2}\zic)=[(\frac{4b}{c+ y^2}-\frac{2}{y^2}-a)\zic+y^{-1}\p_y\zic]\g_{abc}y^3,$$
  we find
\begin{equation}\label{ip.N.zeta.i-est}
 |\ip{\Nab}{\zic}|\lesssim \LpNorm{2}{(c+ y^2)^{-\frac{2-i}{2}}\phi}^2
 + a\LpNorm{2}{(c+ y^2)^{-\frac{1-i}{2}}\phi}^2.
\end{equation}
Estimates \eqref{LXiZeta-est} and \eqref{ip.N.zeta.i-est} give
\begin{equation}\label{Ri-est-c}
 |R_{i}(\phi)|\lesssim (d+a)^{1-i}\LpNorm{2}{\phi}+\LpNorm{2}{(c+ y^2)^{-\frac{2-i}{2}}\phi}^2\,.
\end{equation}
Since we assumed $d\ls a\ln(a^{-1})$, the above estimates imply the following inequalities for $ f_{i}$ and $ r_i$
 \begin{align}\label{fi-ri-est}
|f_i |&\ls a^{i-2} (\ln\frac{1}{a})^{i}\LpNorm{2}{\phi},\\
 |r_i |&\ls a^{i-1} |R_{0}(\phi)|+(a\ln(a^{-1}))^{i} |R_{1}(\phi)| \lesssim  a^{i-1}  [(d+a)\LpNorm{2}{\phi}+\LpNorm{2}{(c+ y^2)^{-1}\phi}^2]\notag\\
 &+(a\ln(a^{-1}))^{i} [\LpNorm{2}{\phi}+\LpNorm{2}{(c+ y^2)^{-\frac{1}{2}}\phi}^2]\notag\\
 & \lesssim  (a\ln(a^{-1}))^{i}\LpNorm{2}{\phi}+a^{i-1}\LpNorm{2}{\phi}^2, \end{align}

The estimates \eqref{Si1-est} 
show that $g=\frac{\ln(a^{-1})}{a}(1+o(1))$, provided
 \begin{align}\label{Sij-est-c}
&a^{-1} |S_{02}|,\  |S_{01}|,\  |S_{02}|,\  |S_{12}|,\ a \ln(a^{-1}) |S_{11}|,\    a |S_{02}|   |S_{11}|,\  a |S_{01}|   |S_{12}| \ll a^{-1}\ln(a^{-1}), \end{align}  
which holds, provided  $ \|\phi\| \ll 1$.  
Therefore $g$ is invertible and its inverse is of the form $g^{-1}=  \frac{a}{\ln(a^{-1})}(1-o(1))$. Hence the equations  \eqref{atauctau-eq} can be rewritten as \eqref{abc-eq1} -- \eqref{abc-eq2}, with 
$  S_{i}(\phi, a, b, c) =\frac{f_i}{g}$ and $ \mathcal{R}_i(\phi, a, b, c)=\frac{r_i}{g}$. Then the estimates of $f_i,\ r_i$ and $g$ given above, imply \eqref{Si} -- \eqref{Rem-i}.

\paragraph{Proof of estimates \eqref{de-est1} - \eqref{de-est4}.} \label{sec:de-est}
Use $\e^{-ay^2/2}=-\frac{1}{ay}\partial_y\e^{-ay^2/2}$ and integration by parts to obtain
\begin{align}
\ip{1}{\zeta_{bc0}}&=\frac{1}{16}\int_0^\infty y(c+y^2)^{2d}\e^{-a\frac{y^2}{2}}  \d y\\
&= \frac{1}{16a}( c^{2d}+4d\int_0^\infty (c+y^2)^{2d-1}y\e^{-a\frac{y^2}{2}}  \d y)\,.
\end{align}
To extract the leading part in the last integral above we rescale $y\to\sqrt{a}y$ to obtain
\begin{align*}
\int_0^\infty (c+y^2)^{2d-1}y\e^{-a\frac{y^2}{2}}  \d y=a^{-2d}\int_0^\infty (ac+y^2)^{2d-1}y\e^{-\frac{y^2}{2}}  \d y\,.
\end{align*}
Next, split the integral up into
\begin{align*}
\int_0^\infty (ac+y^2)^{2d-1}y\e^{-\frac{y^2}{2}}  \d y=\int_0^1 (ac+y^2)^{2d-1}y\e^{-\frac{y^2}{2}}  \d y+\int_1^\infty (ac+y^2)^{2d-1}y\e^{-\frac{y^2}{2}}  \d y\,.
\end{align*}
The second term on the left hand side is uniformly bounded in $a,d$ small, so it suffices to investigate the first term. Write
\begin{align*}
\int_0^1 (ac+y^2)^{2d-1}y\e^{-\frac{y^2}{2}}  \d y=\int_0^1 (ac+y^2)^{2d-1}y\d y+\int_0^1 (ac+y^2)^{2d-1}y(\e^{-\frac{y^2}{2}}-1)  \d y  \,,
\end{align*}
where again the second term is uniformly bounded in $a,d$ small. Explicit integration in the first term yields
\begin{align*}
\int_0^1 (ac+y^2)^{2d-1}y\d y =\frac{1}{4d}((1+ac)^{2d}-(ac)^{2d})  \,.
\end{align*}
By assumption, there is an $\eps>0$c such that $d(a)\leq a^{\eps}$. In particular, $a^{d}\to 1$ as $a\to 0$ so
\begin{align*}
(1+ac)^{2d}-(ac)^{2d}=(1+2d O(ac))-(1+2d O(\ln ac))=O(d\ln\frac{1}{a})\,,
\end{align*}
yielding
\begin{align*}
\ip{1}{\zeta_{bc0}}=\frac{1}{16a}(1+O(d\ln\frac{1}{a}))\,.
\end{align*}
The remaining terms are estimated similarly.
\end{proof}

\bibliographystyle{plain}

\bibliographystyle{plain}

\begin{thebibliography}{10}


\bibitem{AlberChenGlimmLushnikov2006} M. Alber, N. Chen, T. Glimm, and P.M. Lushnikov. Phys. Rev. E.
\textbf{73}, 051901 (2006).

\bibitem{AlberChenLushnikovNewman2007}
M. Alber, N. Chen, P.~M. Lushnikov, and S.~A. Newman.
Physical Review Letters, {\bf 99}, 168102 (2007).


\bibitem{B} W. Beckner, Sharp Sobolev inequalities on the sphere and the Moser-Trudinger inequality, Ann. of Math.,
2, 138 (1993), pp. 213–242.


\bibitem{BeCL} A. L. Bertozzi, J. A. Carillo, Th. Laurent, Blow-up in multidimensional aggregation equations with mildly singular interaction kernels. Nonlinearity 22 (2009) 683 –- 710.



\bibitem{Bi1995}
P. Biler.
\newblock Growth and accretion of mass in an astrophysical model.
\newblock {\em Appl. Math. (Warsaw)}, 23:179-189, 1995

\bibitem{Bi1998}
P. Biler.
\newblock Local and global solvability of some parabolic systems modeling chemotaxis.
\newblock {\em Adv. Math. Sci. Appl.}, 8:715-743, 1998

\bibitem{BiKaLa2006}
P. Biler, G. Karch, and P. Laurencot.
\newblock The $8\pi$-problem for radially symmetric solutions of a chemotaxis model in the plane.
\newblock {\em Preprint}, 2006


\bibitem{BOS} P. Bizo\'n, Yu. N. Ovchinnikov, I. M. Sigal, Collapse of an instanton. Nonlinearity {\bf 17} (2004), no. 4, 1179 –- 1191.


\bibitem{BlDoPe2006}
A. Blanchet, J. Dolbeault, and B. Perthame.
\newblock Two-dimensional Keller-Segel model: Optimal critical mass and qualitative properties of the solutions.
\newblock {\em Electr. J. Diff. Eqns.}, 44:1-33, 2006



\bibitem{BCC} A. Blanchet, E. Carlen, J. A. Carrillo, Functional inequalities, thick tails and asymptotics for the critical mass Patlak-Keller-Segel model.  arXiv:1009.0134.

\bibitem{BCL} A. Blanchet,  J. A. Carrillo,  P. Lauren\c{c}ot, Critical mass for a Patlak-Keller-Segel model with degenerate diffusion in higher dimensions. Calc. Var. Partial Differential Equations 35 (2009), no. 2, 133 –- 168.

\bibitem{BCM}    A. Blanchet,  J. A. Carrillo,  N. Masmoudi, Infinite time aggregation for the critical Patlak-Keller-Segel model in $\RR^2$. Comm. Pure Appl. Math. 61 (2008), no. 10, 1449 –- 1481.

\bibitem{BDEF} A. Blanchet, J. Dolbeault, M. Escobedo, J. Fernandez, Asymptotic behaviour for small mass in the two-dimensional parabolic-elliptic Keller-Segel model. J. Math. Anal. Appl. 361 (2010), no. 2, 533 –- 542.

 \bibitem{BDP}  A. Blanchet,  J. Dolbeault,  B. Perthame,  Two-dimensional Keller-Segel model: optimal critical mass and qualitative properties of the solutions. Electron. J. Differential Equations, No. 44:1-33, 2006.




\bibitem{Bo1967}
J.T. Bonner.
\newblock The cellular slime molds.
\newblock {\em Princeton University Press}, Princeton, New Jersey, second edition, 1967

\bibitem{BrCoKaScVe1999}
M.P. Brenner, P. Constantin, L.P. Kadanoff, A. Schenkel, S.C. Venkataramani.
\newblock Diffusion, attraction and collapse.
\newblock {\em Nonlinearity}, 12:1071-1098, 1999



\bibitem{BrLeBu1998}
M.P. Brenner, L.S. Levitov, and E.O. Budrene.
\newblock Physical mechanisms for chemotactic pattern formation by bacteria.
\newblock {\em Biophys. J.}, 74:1677-1693, 1998



\bibitem{BP}  Buslaev, V.S. and Perel'man, G.S., On the
stability of solitary waves for nonlinear Schr\"odinger equations.
{\it Amer. Math. Soc. Transl. Ser.}, {\bf 2}, 74--98 (1995).

\bibitem{BS} Buslaev, V.S. and Sulem C., On the stability of
solitary waves for nonlinear Schr\"odinger equations,  {\it Ann.
IHP. Analyse Nonline\'eaire}, {\bf 20}, 419--475 (2003).

\bibitem{CL} E. Carlen and M. Loss, Competing symmetries, the logarithmic HLS inequality and Onofri's inequality on
Sn, Geom. Funct. Anal., 2 (1992), pp. 90–104.


\bibitem{CaFi2011}
E. Carlen and A. Figalli.
\newblock Stability for a GNS inequality and the Log-HLS inequality, with application to the critical mass Keller-Segel equation.
\newblock {\em Preprint},  arXiv:1107.5976



\bibitem{CarmelietNatMEd2000} P. Carmeliet, Mechanisms of angiogenesis and arteriogenesis. Nat. Med. {\bf 6}, 389-395 (2000).

\bibitem{CFTV}  Carrillo, José A.; Fornasier, Massimo; Toscani, Giuseppe; Vecil, Francesco Particle, kinetic, and hydrodynamic models of swarming.
Mathematical modeling of collective behavior in socio-economic and life sciences, 297 –- 336, Model. Simul. Sci. Eng. Technol., Birkh\"auser Boston, Inc., Boston, MA, 2010.


\bibitem{ChavSir}  P.-H. Chavanis, C. Sire, Exact analytical solution of the collapse of self-gravitating Brownian particles and bacterial populations at zero temperature. Phys. Rev. E (3) 83 (2011),  031131


\bibitem{ChPe1981}
S. Childress and J.K. Percuss.
\newblock Nonlinear aspects of chemotaxis.
\newblock {\em Math. Bisosc.}, 56:217-237, 1981.


\bibitem{CKT1}  Constantin, P.; Kevrekidis, I. G.; Titi, E. S. Asymptotic states of a Smoluchowski equation. Arch. Ration. Mech. Anal. 174 (2004), no. 3, 365–384.

\bibitem{CKT2}      Constantin, Peter; Kevrekidis, Ioannis; Titi, Edriss S. Remarks on a Smoluchowski equation. Discrete Contin. Dyn. Syst. 11 (2004), no. 1, 101–112.



\bibitem{CFKS}  H.L. Cycon, R.G. Froese , W. Kirsch and B. Simon.  \newblock  Schr\"odinger Operators with application to quantum mechanics and global geometry. {\it Texts and Monographs in Physics. Springer
study edition. Springer-Verlag Berlin} (1987).

\bibitem{DGSW} S.~Dejak, Z.~Gang, I. M. Sigal, S. Wang. \newblock Blowup dynamics in nonlinear heat equations.
 Adv Appl Math 40, 433, 2008.



\bibitem{dlos}
S.I. Dejak, P.M. Lushnikov, Yu.N. Ovchinnikov, and I.M. Sigal
\newblock On spectra of linearized operators for Keller-Segel models of chemotaxis
\newblock {\em Physica D}, 241:1245-1254, 2012.


\bibitem{Doi} Doi, M.: Molecular dynamics and rheological properties of concentrated solutions of rodlike polymers in isotropic and liquid crystalline phases. J. Polym. Sci., Polym. Phys. Ed. 19, 229–243 (1981).



\bibitem{DLV}  S.A. Dyachenko, P.M. Lushnikov and N. Vladimirova. {\it Logarithmic-type Scaling of the Collapse of Keller-Segel Equation}. AIP Conf. Proc. {\bf 1389}, 709-712
   (2011).

\bibitem{GaZa1998}
H. Gajewski and K. Zacharias.
\newblock Global behavior of a reaction-diffusion system modelling chemotaxis.
\newblock {\em Math. Nachr.}, 195:77-114, 1998

\bibitem{GaSig1} Gang Zhou and Sigal, I.M..  On Soliton Dynamics in Nonlinear
Schr\"odinger Equations, {\it GAFA}, {\bf 16} (2006) 1377-1390.



\bibitem{GaSig2} Gang Zhou and Sigal, I.M.. Relaxation of solitons in nonlinear
Schr\"odinger equations with potentials, {\it Adv Math}, 216 (2007),
no. 2, 443--490.

\bibitem{GS} S. Gustafson and I.M. Sigal. \newblock Mathematical Concepts of Quantum Mechanics. \newblock {\em Springer-Verlag}, Berlin, Heidelberg, 2011

\bibitem{HeVe1996a}
M.A. Herrero, and J.J.L. Vel\'{a}zquez.
\newblock Chemptactic collapse for the Keller-Segel model.
\newblock {\em J. Math. Biol.}, 35:177-194, 1996

\bibitem{HeVe1996b}
M.A. Herrero, and J.J.L. Vel\'{a}zquez.
\newblock Singularity patterns in a chemotaxis model.
\newblock {\em Math. Ann.}, 306:583-623, 1996

\bibitem{HeVe1997}
M.A. Herrero, and J.J.L. Vel\'{a}zquez.
\newblock A blowup mechanism for a chemotaxis model.
\newblock {\em Ann. Scuola Norm. Sup. Pisa Cl. Sci.}, 24:633-683, 1997

\bibitem{HeMeVe1997}
M.A. Herrero, E. Medina, and J.J.L. Vel\'{a}zquez.
\newblock Finite-time aggregation into a single point in a reaction-diffusion system.
\newblock {\em Nonlinearity}, 10:1739-1754, 1997

\bibitem{HeMeVe1998}
M.A. Herrero, E. Medina, and J.J.L. Vel\'{a}zquez.
\newblock Self-similar blowup for a reaction-diffusion system.
\newblock {\em J. Comput. Appl. Math.}, 97:99-119, 1998


\bibitem{HP} T. Hillen and K. J. Painter, A user's guide to PDE models for chemotaxis, Journal of Mathematical Biology, 58 (2009), pp. 183–217.

\bibitem{Ho2001}
D. Horstmann.
\newblock The nonsymmetric case of the Keller-Segel model in chemotaxis: some recent results.
\newblock {\em NoDEA Nonlinear Differential Equations Appl.}, 8:399-423, 2001

\bibitem{Ho2002}
D. Horstmann.
\newblock On the existence of radially symmetric blowup solutions for the Keller-Segel model.
\newblock {\em J. Math. Biol.}, 44:463-478, 2002

\bibitem{Ho2003}
D. Horstmann.
\newblock From 1970 until present: The Keller-Segel model in chemotaxis and its consequences. I.
\newblock {\em Jahresber. Deutsch. Math.-Verein}, 105:103-165, 2003

\bibitem{Ho2004}
D. Horstmann.
\newblock From 1970 until present: The Keller-Segel model in chemotaxis and its consequences. II.
\newblock {\em Jahresber. Deutsch. Math.-Verein}, 106:51-69, 2004

\bibitem{HoWa2001}
D. Horstmann and G. Wang.
\newblock Blowup in a chemotaxis model without symmetry assumptions.
\newblock {\em European J. Appl. Math.}, 12:159-177, 2001

\bibitem{HoWi2005}
D. Horstmann and M. Winkler.
\newblock Boundedness vs. blowup in a chemotaxis system.
\newblock {\em J. Differential Equations}, 215:52-107, 2005

\bibitem{JaLu1992}
W. J\"ager and S. Luckhaus.
\newblock On explosions of solutions to a system of partial differential equations modelling chemotaxis.
\newblock {\em Trans. Amer. Soc.}, 329:819-824, 1992


\bibitem{KeSe1970}
E.F. Keller and L.A. Segel.
\newblock Initiation of slime mold aggregation viewed as an instability.
\newblock {\em J. Theor. Biol.}, 26:300-415, 1970



\bibitem{KST1} J. Krieger, W. Schlag and D. Tataru, Renormalization and blow up for charge one equivariant critical wave maps, Invent. Math. 171 (2008), no. 3, 543 –- 615.

\bibitem{KST2} J. Krieger, W. Schlag and D. Tataru, Renormalization and blow up for critical Yang-Mills problem, e-print, arXiv:0809.211, 2008.

\bibitem{Lar} Larson, R.G.: The Structure and Rheology of Complex Fluids. Oxford University Press,
London, 1999

\bibitem{Lush} P.M. Lushnikov. {\it Critical chemotactic collapse}. Physics Letters A, {\bf 374}, 1678-1685   (2010).


\bibitem{LushnikovChenalberPRE2008}
P.~M. Lushnikov, N.~Chen, and M.~Alber.
Physical Review E,  {\bf 78}, 061904  (2008).



\bibitem{MR} F. Merle and   P. Rafael.
  \newblock   Blow up of the critical norm for some radial $L2$ super critical nonlinear Schrodinger equations. Amer. J. Math. 130, 945(2008).

\bibitem{MZ} F. Merle and H.  Zaag,  Blow-up behavior outside the origin for a semilinear wave equation in the radial case. Bull. Sci. Math. 135 (2011), no. 4, 353 –- 373

\bibitem{Na1973}
V. Nanjundiah.
\newblock Chemotaxis, signal relaying and aggregation morphology.
\newblock {\em J. Theor. Biol.}, 42:63-105, 1973

\bibitem{Na1995}
T. Nagai.
\newblock Blow-up of radially symmetric solutions to a chemotaxis system.
\newblock {\em Adv. Math. Sci. Appl.}, 5:581-601, 1995

\bibitem{Na2001}
T. Nagai.
\newblock Blow-up of nonradial solutions to parabolic-elliptic systems modeling chemotaxis.
\newblock {\em J. Inequal. Appl.}, 6:37-55, 2001

\bibitem{Na20012}
T. Nagai.
\newblock Global existence and blowup of solutions to a chemotaxis system.
\newblock {\em Proceedings of the Third World Congress of Nonlinear Analysis, Part2 (Catania, 2000)}, 47:777-787, 2001

\bibitem{NaSe1997}
T. Nagai and T. Senba.
\newblock Behavior of radially symmetric solutions of a system related to chemotaxis.
\newblock {\em Proceedings of the Third World Congress of Nonlinear Analysis, Part 6 (Athens, 1996)}, 30:3837-3842, 1997

\bibitem{NaSe1998}
T. Nagai and T. Senba.
\newblock Global existence and blowup of radial solutions to a parabolic-elliptic system of chemotaxis.
\newblock {\em Adv. Math. Sci. Appl.}, 8:145-156, 1998

\bibitem{NaSeSu2000}
T. Nagai, T. Senba, and T. Suzuki.
\newblock Chemotactic collapse in a parabolic system of mathematical biology.
\newblock {\em Hiroshima Math. J.}, 30:463-497, 2000

\bibitem{NaSeYo1997}
T. Nagai, T. Senba, and K. Yoshida.
\newblock Application of the Trudinger-Moser inequality to a parabolic system of chemotaxis.
\newblock {\em Funkcial. Ekvac.}, 40:411-433, 1997


\bibitem{NewmanGrima2004} T. J. Newman, and R. Grima. Phys. Rev. E   {\bf 70},
051916 (2004).


\bibitem{Oe1989}
K. Oelschl\"ager.
\newblock On the derivation of reaction-diffusion equations as limit dynamics of systems of moderately interacting stochastic processes.
\newblock {\em Probab. Theory Related Fields}, 82:565-586, 1989


\bibitem{Ost} J. Ostriker, Ap J 140 10560, 1964.


\bibitem{OtSt1997}
H.G. Othmer and A. Stevens.
\newblock Aggregation, blowup, and collapse: the ABCs of taxis in reinforced random walks.
\newblock {\em SIAM J. Appl. Math.}, 57:1044-1081, 1997


\bibitem{OS}  Yu. N. Ovchinnikov, I. M. Sigal, On Collapse of Wave Maps. Physica D 240 (2011), pp. 1311 –- 1324.


\bibitem{Patlak1953} C.~S. Patlak.
\newblock Random walk with persistence and external bias.
  Bull. Math. Biophys {\bf 15}, 311 (1953).



\bibitem{Per} B. Perthame, Transport equations in biology, Frontiers in Mathematics, Birkh¨auser Verlag, Basel, 2007.


\bibitem{RR}
P.Rafael and I.Rodnianski.
  \newblock Stable blow up dynamics for the critical co-rotational Wave Maps and equivariant Yang-Mills problems.
  arXiv:0911.0692, 2010.


\bibitem{ReSiIV} M. Reed and B. Simon. \newblock Methods of Modern Mathematical Physics IV. \newblock {\em Academic Press}, San Diego, California, 1979


\bibitem{RS} I. Rodnianski and J.Sternbenz, On the formation of singularities in the critical $O(3)\ \sigma$-model,  Ann. of Math. (2) 172 (2010), no. 1, 187  –-  242.


\bibitem{SirChav} C. Sire. P.-H. Chavanis,  Critical dynamics of self-gravitating Langevin particles and bacterial populations. Phys. Rev. E (3) 78 (2008),  061111


\bibitem{SW1} Soffer, A. and Weinstein, I.M. Multichannel
nonlinear scattering for nonintegrable equations.  {\it Comm. Math.
Phys.}, {\bf 133}, 119--146 (1990).

\bibitem{SW2} Soffer, A. and Weinstein, I.M.  Multichannel
nonlinear scattering for nonintegrable equations II.  The case of
anisotropic potentials and data. {\it J. Diff. Equations}, {\bf 98},
376--390 (1992).

\bibitem{SW3} Soffer, A. and Weinstein, I.M., Selection of the
ground state for nonlinear Schr\"odinger equations. Rev. Math. Phys,
ArXiv:nlin.PS/0308020 (2003).






\bibitem{St2000}
A. Stevens.
\newblock The derivation of chemotaxis equations as limit dynamics of moderately interacting stochastic many-particle systems.
\newblock {\em SIAM J. Appl. Math.}, 61:183-212 (electronic), 2000

\bibitem{Str} M. Struwe, Equivariant wave maps in two space dimensions. Comm. Pure Appl. Math. 56 (2003), 815-823.


\bibitem{SS} C. Sulem and P.L. Sulem, Nonlinear Schr\"odinger Equation. \textit{Series in Mathematical Sciences, Volume 139, Springer-Verlag}, 1999.


\bibitem{TY1} Tsai, T.-P. and Yau, H.-T., Asymptotic dynamics
of nonlinear Schr\"odinger equations: resonance-dominated and
dispersion-dominated solutions.  {\it Comm. Pure Appl. Math.}, {\bf
55}, 153--216 (2002).

\bibitem{TY2} Tsai, T.-P. and Yau, H.-T., Relaxation of
excited states in nonlinear Schr\"odinger Equations, {\it Int. Math.
Res. Not.}, {\bf 31}, 1629--1673 (2002).

\bibitem{TY3} Tsai, T.-P. and Yau, H.-T., Stable directions
for excited states of nonlinear Schr\"odinger equations, {\it Comm.
PDE}, {\bf 27}, 2363--2402 (2002).

\bibitem{Ve1}
J.~J.~L. Vel{\'a}zquez,
\newblock Stability of some mechanisms of chemotactic aggregation.
\newblock {\em SIAM J. Appl. Math.}, 62(5):1581 –- 1633 (electronic), 2002.

\bibitem{Ve2}
J.~J.~L. Vel{\'a}zquez, \newblock {\em SIAM J. Appl. Math.}, 64(4):1198 –- 1248, 2004.


\bibitem{Wo1992}
G. Wolansky.
\newblock On steady distributions of self-attracting clusters under friction and fluctuations.
\newblock {\em Arch. Rational Mech. Anal.}, 119:355-391, 1992

\end{thebibliography}

\end{document}